\documentclass[journal,twoside,web]{ieeecolor}
\pdfoutput=1
\usepackage{generic}
\usepackage{graphicx}
\usepackage{cite}
\usepackage{amsmath,amssymb,amsfonts}
\usepackage{mathrsfs}

\usepackage{amsthm}
\newtheorem{theorem}{Theorem}
\theoremstyle{definition} %% Define the style of environment
\newtheorem{definition}{Definition}
\newtheorem{lemma}{Lemma}

\newtheorem{proposition}{Proposition}
\newtheorem{remark}{Remark}
\newtheorem{assumption}{Assumption}

\usepackage{mathtools} %% Matrix description

\usepackage{nicematrix}
\usepackage{algorithm}
\usepackage{algorithmicx}
\usepackage{algpseudocode}
\makeatletter
\newenvironment{breakablealgorithm}
  {% \begin{breakablealgorithm}
   \begin{center}
     \refstepcounter{algorithm}% New algorithm
     \hrule height.8pt depth0pt \kern2pt% \@fs@pre for \@fs@ruled
     \renewcommand{\caption}[2][\relax]{% Make a new \caption
       {\raggedright\textbf{Smart enumeration algorithm (SEA)\thealgorithm} ##2\par}%
       \ifx\relax##1\relax % #1 is \relax
         \addcontentsline{loa}{algorithm}{\protect\numberline{\thealgorithm}##2}%
       \else % #1 is not \relax
         \addcontentsline{loa}{algorithm}{\protect\numberline{\thealgorithm}##1}%
       \fi
       \kern2pt\hrule\kern2pt
     }
  }{% \end{breakablealgorithm}
     \kern2pt\hrule\relax% \@fs@post for \@fs@ruled
   \end{center}
  }
\makeatother
\usepackage{float}
\usepackage[section]{placeins} % Keep the figure in the current section
\usepackage{subfig}
\usepackage{overpic}

\let\labelindent\relax
\usepackage{enumitem} %% Adjust the indentation
\usepackage{booktabs} %% Three-line table
\usepackage{makecell} %% Line break within a table cell

\makeatletter
\let\NAT@parse\undefined
\makeatother
\usepackage{hyperref}
\hypersetup{hidelinks=true}
\usepackage{textcomp}
\def\BibTeX{{\rm B\kern-.05em{\sc i\kern-.025em b}\kern-.08em
    T\kern-.1667em\lower.7ex\hbox{E}\kern-.125emX}}

\begin{document}
\title{An Online Cross-layered Defense Strategy with Bandwidth Allocation for Multi-channel Systems under DoS Attacks}

\author{Liheng~Wan,~Panshuo~Li,~James~Lam% <-this % stops a space
\thanks{The work was partially supported by the National Natural Science Foundation of China under Grants 62473101, 62121004 and 62273286, Committee on Research and Conference Grants under Grant 2202100888, and Guangzhou Basic and Applied Basic Research Foundation under Grant SL2024A04J00332. (Corresponding authors: Panshuo Li; James Lam)}% <-this % stops a space
\thanks{L.~Wan and P.~Li are with the School of Automation and the Guangdong Provincial Key Laboratory of Intelligent Decision and Cooperative Control, Guangdong University of Technology, and also with the Guangdong-Hong Kong Joint Laboratory for Intelligent Decision and Cooperative Control, Guangzhou 510006, China (e-mail: alexwan2020@163.com; panshuoli812@gdut.edu.cn)}
\thanks{J.~Lam is with the Department of Mechanical Engineering, The University of Hong Kong, Hong Kong, and also with the Guangdong-Hong Kong Joint Laboratory for Intelligent Decision and Cooperative Control, Guangzhou 510006, China (e-mail: james.lam@hku.hk)}% <-this % stops a space
}

\maketitle

\begin{abstract}
This paper proposes an online cross-layered defense strategy for multi-channel systems with switched dynamics under DoS attacks. The enabling condition of a channel under attacks is formulated with respect to attack flow and channel bandwidth, then a new networked control system model bridging the gap between system dynamics and network deployment is built. Based on this, the cross-layered defense strategy is proposed. It jointly optimizes the controller gain and bandwidth allocation of channels according to the real-time attack flow and system dynamics, by solving a mixed-integer semidefinite programming online. A smart enumeration algorithm for non-convex bi-level optimization is proposed to analyze the stability under the strategy. Numerical examples are given to illustrate the high resilience from the cross-layered feature.
\end{abstract}


\section{Introduction}
Modern industrial systems contain various variables, such as temperature and humidity, which are of concern to researchers and engineers. These variables are measured by different sensors and transmitted through different channels respectively, which brings the rapid development of network communication. However, the introduction of networks poses the risk of DoS attacks, which are highly destructive with a simple mechanism \cite{network}: The attacker inputs attack flow overwhelming the bandwidth of a channel to jam the channel buffer, which causes the dropouts of incoming packets.

For networked control systems, existing studies \cite{review-Annual-2019,review-CJAS-2022,review-Annual-2023} tend to formulate the temporal distribution of attack-induced packet dropouts directly. In \cite{res-Italy-2015}, the intensity of DoS attacks is characterized with constraints on frequency and duration of attacks in any time interval, which is the mainstream framework of the researches on DoS attacks in recent years \cite{deter-Wang-2022,deter-Zhao-2022,deter-Xu-2024}. Furthermore, the attack intensity can be described with the maximum number of consecutive attacks, especially in the predictive control scheme \cite{active-2021,active-2023}. Apart from the above deterministic formulations, the packet dropouts sequence under DoS attacks can be modeled as a Bernoulli process \cite{Bern-Li-2024,Bern-TAC-2025} or a Markov chain \cite{res-Sun-2017,Markov-Xue-2025}.

Based on the above models, the controller design under DoS attacks is widely studied. Resilient stabilizing control with a sampling rate adapting to attack intensity is studied in \cite{res-Italy-2015}. And the controller gain adapting to the number of consecutive packet dropouts is explored in \cite{res-Sun-2017}. Moreover, the predictive control is investigated in \cite{active-2021,active-2023} to update control input at instants under attacks by utilizing the previous sampled data. Furthermore, considering the switched system dynamics, the control for switched systems under DoS attacks is studied in \cite{deter-Wang-2022,deter-Zhao-2022,active-2023}. On the other hand, the optimal DoS attack scheduling is studied in \cite{Optimal_DoS_1} and \cite{Optimal_DoS_2}.

One may observe that tremendous results have been reported on the controller design under DoS attacks. However, recalling the nature of DoS attacks mentioned earlier, there appears to be significant potential for performance enhancement through the synthesis of network deployment and controller design. However, there is only a scarcity of research dedicated to this area. In \cite{bandwidth-limit-conf-2018,bandwidth-limit-Hossain-2022}, considering a group of interconnected systems, bandwidth is allocated according to whether the system is subject to DoS attacks or not. The impact of bandwidth allocation on the input delay under redundancy control is studied in \cite{bandwidth-redundant-Hu-2021}. It is worth mentioning that the dynamic characteristics of systems are not utilized for bandwidth allocation in the above-mentioned works. The bandwidth allocation and controller gain are regulated sequentially rather than being jointly optimized towards the global optimum.

Motivated by the above discussions, we consider a multi-channel communication framework, where both the attacker and the defender have limited resources. The novelties and contributions can be summarized as:

\begin{enumerate}
\renewcommand{\labelenumi}{\theenumi)}
\item A state-space model embedded by a logical model is built to bridge the gap between network deployment and system dynamics. The logical model characterizes the dynamics of channels' on-off state under attack flow.
\item An online cross-layered defense strategy is proposed to jointly optimize the bandwidth allocation and controller gain to stabilize the system under attacks according to the real-time attack flow and system dynamics, by solving a mixed-integer semidefinite programming online.
\item A smart enumeration algorithm (SEA) is proposed to solve a non-convex bi-level mixed-integer semidefinite programming exactly. The result is employed to analyze the stability under the proposed defense strategy.
\end{enumerate}

\textbf{Notations: }We use $A\backslash B$ to denote the set $\{x\in A:x\notin B\}$, and $\{a,\ldots,b\}$ to denote the discrete-time interval $\{a,a+1,a+2,\ldots,b-1,b\}$. Define $\|\cdot\|$ as the Euclidean vector norm. We use $\mathbf{1}_{m\times n}$ to represent an $m$-by-$n$ matrix of which the entries equal $1$, and $\mathrm{diag}(v)$ to represent a diagonal matrix with a main diagonal vector $v$. We denote $\bar{\lambda}(\cdot)$ as the largest eigenvalue of a real symmetric matrix, and $\underline{\lambda}(\cdot)$ as the smallest one. Moreover, we use $P\succ 0$ ($P\succeq 0$) to indicate that $P$ is a real symmetric and positive (semi)definite matrix, and $A=[a_{ij}]_{m\times n}>0$ ($A\geq 0$) to indicate that each entry $a_{ij}>0$ ($a_{ij}\geq 0$). The symbol `*' denotes the ellipsis entry of a symmetric matrix.

\section{Preliminaries and problem formulation}
\subsection{System framework}\label{section-system-framework}
To illustrate the strategy's adaptability to system dynamics, periodic piecewise linear systems (PPLSs) \cite{PPLS-Auto-2015,PPLS-Auto-2018,PPLS-TAC-2019,PPLS-JFI-2022,PPLS-almost-2023}, a special type of switched systems whose system modes switch cyclically with fixed dwell-times, are taken as the controlled plant. Consider a multi-channel PPLS composed of $s$ subsystems. For $k \in \{\ell T+k_{i-1},\ell T+k_{i-1}+1,\ldots,\ell T+k_{i}-1\}$, $\ell\in\{0,1,\ldots\}$, $i \in \mathcal{S} \triangleq\{1,2,\ldots,s\}$,
\begin{equation}\label{PPLS}
\begin{aligned}
x(k+1)&=A_{i}x(k)+B_{i}u(k),\\
u(k)&=K_{i}(k)L_{i}(k)x(k),\\
L_{i}(k)&=\mathrm{diag}\Big(l_{i}^{(1)}(k),l_{i}^{(2)}(k),\ldots,l_{i}^{(n)}(k)\Big).
\end{aligned}
\end{equation}
where
\begin{equation}
l_{i}^{(j)}(k)\triangleq
\begin{cases}
0, & \text{the channel of $x^{(j)}$ is jammed at time $k$,}\\
1, & \text{otherwise.}
\end{cases}\\
\end{equation}
The channels' state of the $i$th subsystem $(A_{i}, B_{i})$ is represented by $L_{i}(k)$. The system state is denoted by $x(k)=[x^{(1)}(k),x^{(2)}(k),\ldots,x^{(n)}(k)]^{\mathrm{T}}\in \mathbb{R}^{n}$ of which different entries $x^{(j)}(k)$ are measured by different sensors and transmitted through different channels, respectively. Superscript $(j)$ represents both the index of state entries and that of channels. And $u(k)\in \mathbb{R}^{n_u}$ is the control input. Denote $T$ as the fundamental period of the system, $\ell T+k_{i}$ the switching instant when the $i$th subsystem switches to the $(i+1)$th in the $(\ell+1)$th period. Given $k_{0}=0$. The dwell-time of the $i$th subsystem is $T_{i}=k_{i}-k_{i-1}$, and it holds $T=\sum_{i=1}^{s}T_{i}$.

\begin{figure}[htb]
\centering
\includegraphics[width=.45\textwidth]{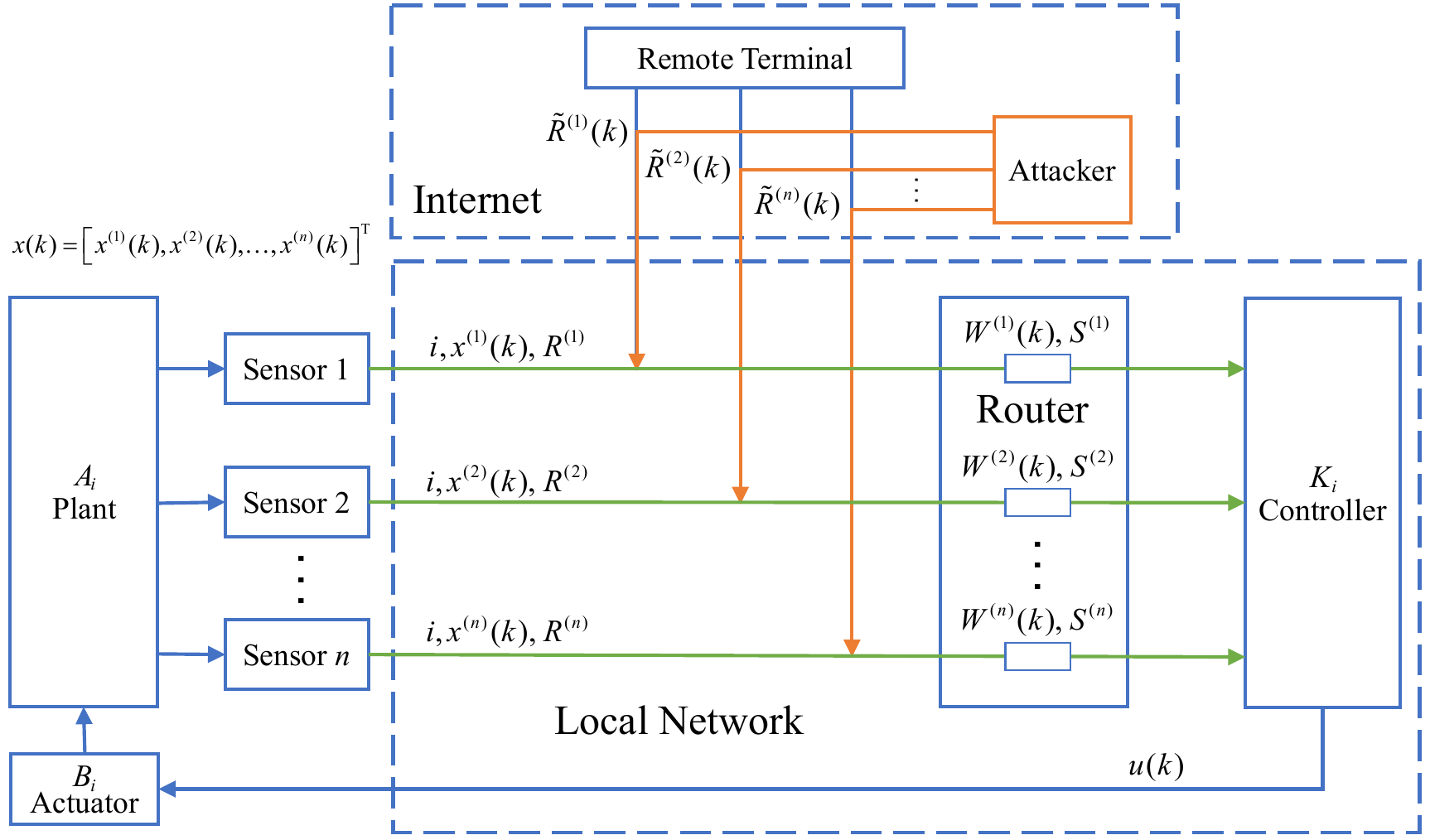}
\caption{Illustration of multi-channel networked control systems}
\label{network}
\end{figure}

As shown in Fig.~\ref{network}, the communication network can be divided as two parts:
\begin{enumerate}\renewcommand{\labelenumi}{\theenumi)}
\item In the Internet, the remote terminal interacts with the controlled plant via the branch channels of the sampling channels (the green lines in Fig.~\ref{network}). However, the branch channels would be also utilized by the attacker in the Internet to input DoS attack flow into the sampling channels.
\item In the local network, the state entries $x^{(j)}(k)$ and system mode $i$ are measured by different sensors, and transmitted in packets through different sampling channels towards the router, respectively. Then, the router transfers the received packets towards the controller. The control input $u(k)$ is transmitted through the reliable channels free from attacks.
\end{enumerate}

\subsection{Enabling condition of sampling channels}\label{section-enable}
As shown in (\ref{PPLS}), the channels' state $L_{i}(k)$ affects the dynamics of the multi-channel system directly. Hence, it is necessary to discuss the enabling condition of sampling channels under DoS attacks. Firstly, some critical terms need to be clarified.
\begin{itemize}
  \item \textbf{Transmission rate:} The number of bits that can be transmitted per unit time.
  \item \textbf{Flow:} The transmission rate of input data. \textbf{Normal flow} $R=[R^{(1)},R^{(2)},\ldots,R^{(n)}]$ is the transmission rate of the packets loaded with state entries and system modes, where $R^{(j)}$ is the normal flow for $x^{(j)}(k)$ and system mode $i$ through channel $j$. \textbf{Attack flow} $\tilde{R}(k)=[\tilde{R}^{(1)}(k),\tilde{R}^{(2)}(k),\ldots,\tilde{R}^{(n)}(k)]$ is the transmission rate of malicious packets at time $k$, where $\tilde{R}^{(j)}(k)$ is the attack flow through channel $j$. Define $R+\tilde{R}(k)$ as \textbf{input flow.} For the defender, $R$ is a known parameter, input flow $R+\tilde{R}(k)$ can be detected in real-time, hence the time-varying attack flow $\tilde{R}(k)$ is accessible at each time $k$.
  \item \textbf{Bandwidth:} The maximum outgoing transmission rate of a channel. Under the limited total available bandwidth $W_{\Sigma}$ such that $0<\sum_{j=1}^{n}W^{(j)}(k)\leq W_{\Sigma}$, the \textbf{bandwidth allocation} is denoted by $W(k)=[W^{(1)}(k),W^{(2)}(k),\ldots,W^{(n)}(k)]$ where $W^{(j)}(k)$ is the bandwidth allocated by the router to channel $j$ at time $k$.
  \item \textbf{Buffer:} The space used to store the packets that cannot be transferred instantly by the router when input flow overwhelms bandwidth. We use $S^{(j)}$ to denote the buffer of channel $j$.
\end{itemize}
Regarding the capacity of the attacker and that of the defender, some assumptions are given as follows considering practical capacity.
\begin{assumption}\label{assump-flow}
Both the total available attack flow and the attack flow for an individual channel are bounded, that is,
\begin{equation}\label{attack-capacity}
0\leq\sum_{j=1}^{n}\tilde{R}^{(j)}(k)\leq \tilde{R}_{\Sigma},\ 0\leq\tilde{R}^{(j)}(k)\leq \bar{R}^{(j)},
\end{equation}
where $\tilde{R}^{(j)}(k)$ and $\bar{R}^{(j)}$ are known to the defender.
\end{assumption}

\begin{assumption}\label{PPLS-assump-intensity}
The total attack duration $\tilde{T}_{\ell,i}$ during a dwell-time $\{\ell T+k_{i-1},\ldots,\ell T+k_{i-1}+T_{i}-1\}$ of the $i$th subsystem is upper-bounded by $\tilde{T}_{i}$. That is, for $\ell\in\{0,1,\ldots\}$, $i=\{1,2,\ldots,s\}$, it holds $\tilde{T}_{\ell,i}\leq \tilde{T}_{i}\leq T_{i}$. Define $\delta_{i}\triangleq\frac{\tilde{T}_{i}}{T_{i}}$ as the upper bound of attack duration ratio for the $i$th subsystem.
\end{assumption}

\begin{assumption}\label{assump-one}
For any $j\in \{1,2,\ldots,n\}$, the total available bandwidth $W_{\Sigma}$ satisfies $W_{\Sigma}>\bar{R}^{(j)}+R^{(j)}$.
\end{assumption}

Then with Fig.~\ref{network}, the impact of DoS attacks on packet transmission can be demonstrated: The attacker inputs overwhelming attack flow $\tilde{R}^{(j)}(k)$ into the channel. If the input flow $R^{(j)}+\tilde{R}^{(j)}(k)$ overwhelms $W^{(j)}(k)$ persistently, the buffer will get jammed and the incoming packets of $x^{(j)}(k)$ will be rejected by the router, that is, the packet dropouts of $x^{(j)}(k)$ occur \cite{network}. Hence, the enabling conditions can be given as follows:
\begin{enumerate}\renewcommand{\labelenumi}{\theenumi)}
\item Allocation condition: The enabled channel should be allocated with enough bandwidth to mitigate the attack flow. It can be established as
    \begin{equation}\label{PPLS-cond-allocation}
    W^{(j)}(k)\geq R^{(j)}+\tilde{R}^{(j)}(k).
    \end{equation}
\item Delay condition: The buffer of the channel should not have been jammed until the newly allocated bandwidth get configured. It can be formulated as
    \begin{equation}\label{PPLS-cond-delay}
    R^{(j)}+\tilde{R}^{(j)}(k)-W^{(j)}(k-1)< \frac{S^{(j)}}{\tau},
    \end{equation}
    where $\tau$ is allocation delay much smaller than sampling period $T_{d}$. Converting the unit discrete-time into a continuous-time sampling period $T_{d}$ as Fig.~\ref{period}, there exists allocation delay $\tau$ laid in the start of the sampling period. At time $t_{k}$, the buffer is cleared up since the packets of $x(t_{k-1})$ are no longer useful. Bandwidth remains $W^{(j)}(k-1)$ in $[t_{k},t_{k}+\tau]$, and $W^{(j)}(k)$ is activated in $[t_{k}+\tau,t_{k}+T_{d}]$. Delay condition (\ref{PPLS-cond-delay}) guarantees that the buffer would not be jammed in $[t_{k},t_{k}+\tau]$ under $W^{(j)}(k-1)$. And allocation condition (\ref{PPLS-cond-allocation}) implies that the free space in buffer is non-decreasing in $[t_{k}+\tau,t_{k}+T_{d}]$ under $W^{(j)}(k)$.
\end{enumerate}
To sum up, the enabling condition can be formulated as
\begin{equation}\label{PPLS-cond}
l_{i}^{(j)}(k)=\begin{cases}
1, & R^{(j)}+\tilde{R}^{(j)}(k)-W^{(j)}(k-1)< \frac{S^{(j)}}{\tau}\\
  & \text{and}\ W^{(j)}(k)\geq R^{(j)}+\tilde{R}^{(j)}(k),\\
0, & \text{otherwise}.
\end{cases}
\end{equation}

\begin{figure}[htb]
\centering
\includegraphics[width=.45\textwidth]{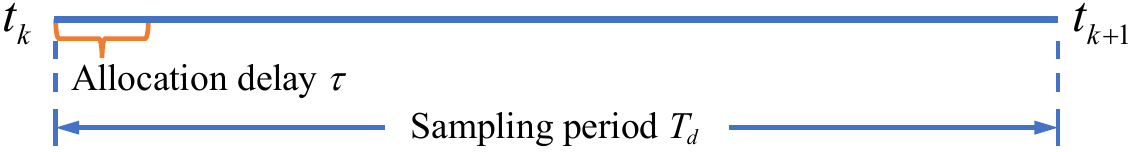}
\caption{Illustration of sampling period}
\label{period}
\end{figure}

\begin{remark}\label{remark-one}
Under Assumption \ref{assump-one}, at least one sampling channel can be enabled under DoS attacks. As mentioned in Section \ref{section-system-framework}, system mode $i$ is transmitted duplicately through all the sampling channels. Hence, system mode $i$ can be transmitted successfully through at least one channel.
\end{remark}

\section{Online cross-layered defense strategy}\label{section-opt}
In this section, the online cross-layered defense strategy is proposed. It optimizes the bandwidth allocation $W(k)$ and controller gain $K_{i}(k)$ jointly according to the detectable real-time attack flow $\tilde{R}(k)$ and system mode $i$. The online optimization is formulated as a mixed-integer semidefinite programming.
\subsection{Online cross-layered optimization}
For a control system, stability and convergence rate are most concerned. Motivated by it, exponential stability is studied in this paper.

\begin{definition}\cite{exponential_stability}
System (\ref{PPLS}) is said to be $\chi$-exponentially stable, if there exist some scalars $c>0 $, $0<\chi<1$, such that the system state satisfies $\|x(k)\|\leq c \chi^{k-k_{0}}\left\|x\left(k_{0}\right)\right\|$, $k\geq k_{0}$.
\end{definition}

For stability analysis, a Lyapunov function with a time-varying Lyapunov matrix \cite{PPLS-JFI-2022} is constructed as
\begin{equation}\label{PPLS-Lya}
\begin{aligned}
  V(k)&=V_{i}(k)=x^{\mathrm{T}}(k)P_{i}(k)x(k),
  \\P_{i}(k)&=P_{i-1} + \frac{k-\ell T - k_{i-1}}{T_{i}}(P_{i}-P_{i-1}),
\end{aligned}
\end{equation}
where $k\in \{\ell T+k_{i-1},\ldots,\ell T+k_{i}-1\}$, $P_{s}=P_{0}$, $P_{s+1}=P_{1}$.

The exponential order of the Lyapunov function is taken as the objective value for the online optimization detailed later. The following lemmas discuss the exponential order of $V(k)$ in the absence and presence of DoS attacks, respectively.

\begin{lemma}\label{PPLS-lemma-alpha}
Consider system (\ref{PPLS}) in the absence of attacks ($L_{i}(k)=I$), given $\alpha_{i}>0$, $i=1,2,\ldots,s$. If there exist matrices $P_{i}\succ 0$ and $G_{i}$, $Y_{i}$ with appropriate dimensions, such that
\begin{equation}\label{PPLS-lemma-alpha-LMI-1}
\begin{bmatrix}
 \frac{T_{i}+1}{T_{i}}\alpha_{i}(P_{i-1}^{-1}-G_{i}^{\mathrm{T}}-G_{i})  & * & *\\
 A_{i}G_{i}+B_{i}Y_{i} & -P_{i-1}^{-1} & *\\
 G_{i} & 0 & -\frac{T_{i}}{\alpha_{i}}P_{i}^{-1}
\end{bmatrix}\prec 0,
\end{equation}
\begin{equation}\label{PPLS-lemma-alpha-LMI-2}
\begin{bmatrix}
 \frac{T_{i}-1}{T_{i}}\alpha_{i}P_{i}^{-1}+\frac{\alpha_{i}}{T_{i}}P_{i-1}^{-1}-\alpha_{i}(G_{i}^{\mathrm{T}}+G_{i}) & *\\
 A_{i}G_{i}+B_{i}Y_{i} & -P_{i}^{-1}
\end{bmatrix}\prec 0,
\end{equation}
then it holds $V_{i}(k+1)\leq \alpha_{i}V_{i}(k)$ under controller gain $K_{i}=Y_{i}G_{i}^{-1}$.\hfill $\Box$
\end{lemma}

\begin{proof}
Inequalities (\ref{PPLS-lemma-alpha-LMI-1}) and (\ref{PPLS-lemma-alpha-LMI-2}) align with (21) and (22) in \cite{PPLS-JFI-2022}, which guarantees $V_{i}(k+1)\leq \alpha_{i}V_{i}(k)$ according to (9)--(15) in \cite{PPLS-JFI-2022}.
\end{proof}

\begin{lemma}\label{PPLS-lemma-beta}
Consider system (\ref{PPLS}) under attacks, given $P_{i}\succ 0$ obtained from Lemma \ref{PPLS-lemma-alpha}, $i=1,2,\ldots,s$. If there exist $\beta_{i}(k)>0$ and matrices $L_{i}(k)$, $K_{i}(k)$, such that
\begin{equation}\label{PPLS-lemma-beta-LMI-1}
\begin{bmatrix}
-\frac{T_{i}+1}{T_{i}}\beta_{i}(k)P_{i-1}+\frac{\beta_{i}(k)}{T_{i}}P_{i} & *\\
A_{i}+B_{i}K_{i}(k)L_{i}(k) & -P_{i-1}^{-1}
\end{bmatrix}\preceq 0,
\end{equation}
\begin{equation}\label{PPLS-lemma-beta-LMI-2}
\begin{bmatrix}
-\frac{T_{i}-1}{T_{i}}\beta_{i}(k)P_{i}-\frac{\beta_{i}(k)}{T_{i}}P_{i-1} & *\\
A_{i}+B_{i}K_{i}(k)L_{i}(k) & -P_{i}^{-1}
\end{bmatrix}\preceq 0,
\end{equation}
then it holds $V_{i}(k+1)\leq\beta_{i}(k)V_{i}(k)$.\hfill $\Box$
\end{lemma}
\begin{proof}

The above result can be obtained from (12)--(15) of \cite{PPLS-JFI-2022} by substituting $A_{i}$ there with $A_{i}+B_{i}K_{i}(k)L_{i}(k)$.
\end{proof}

When free from attacks, the controller gain is obtained with Lemma \ref{PPLS-lemma-alpha}, which also provides the Lyapunov matrices $P_{i}$, $i=1,2,\ldots,s$, for the online optimization (\ref{PPLS-opt}) under attacks.

When subject to attacks at time $k$, we optimize the bandwidth allocation $W(k)$ and controller gain $K_{i}(k)$ jointly to minimize the exponential order $\beta_{i}(k)$ such that $V_{i}(k+1)\leq \beta_{i}(k)V_{i}(k)$. It can be formulated as the following mixed-integer semidefinite programming, where the attack flow $\tilde{R}(k)$ serves as a time-varying parameter, the bandwidth allocation $W(k)$ and controller gain $K_{i}(k)$ serve as decision variables.
\begin{subequations}\label{PPLS-opt}
\begin{align}
&\beta_{i}^{*}(k)=\min\limits_{W(k),L_{i}(k),K_{i}(k),\beta_{i}(k)}\ \beta_{i}(k)\label{PPLS-opt-obj}\\
\text{s.t.}\quad &l_{i}^{(j)}(k)=\fontsize{9}{11}\selectfont{\begin{cases}
1, & \big(R^{(j)}+\tilde{R}^{(j)}(k)-W^{(j)}(k-1)\big)\tau< S^{(j)}\\
  & \text{and}\ W^{(j)}(k)\geq R^{(j)}+\tilde{R}^{(j)}(k),\\
0, & \text{otherwise}.
\end{cases}}\label{PPLS-opt-logic}\\
&\sum_{j=1}^{n}W^{(j)}(k)\leq W_{\Sigma},\ W^{(j)}(k)\geq 0,\ \sum_{j=1}^{n} l_{i}^{(j)}(k)\geq 1,\label{PPLS-opt-source}\\
&\begin{bmatrix}
-\frac{T_{i}+1}{T_{i}}\beta_{i}(k)P_{i-1}+\frac{\beta_{i}(k)}{T_{i}}P_{i} & *\\
A_{i}+B_{i}K_{i}(k)L_{i}(k) & -P_{i-1}^{-1}
\end{bmatrix}\preceq 0,\label{PPLS-opt-LMI-1}\\
&\begin{bmatrix}
-\frac{T_{i}-1}{T_{i}}\beta_{i}(k)P_{i}-\frac{\beta_{i}(k)}{T_{i}}P_{i-1} & *\\
A_{i}+B_{i}K_{i}(k)L_{i}(k) & -P_{i}^{-1}
\end{bmatrix}\preceq 0,\label{PPLS-opt-LMI-2}
\end{align}
\end{subequations}
where $L_{i}(k)$, $\tilde{R}(k)$ and $W(k)$ are defined in Section \ref{section-enable}. With Lemma \ref{PPLS-lemma-beta}, $V_{i}(k+1)\leq \beta_{i}^{*}(k)V_{i}(k)$ can be achieved under the optimal $W(k)$ and $K_{i}(k)$ from (\ref{PPLS-opt}).

\begin{remark}
In the online optimization (12), the bandwidth allocation $W(k)$, controller gain $K_{i}(k)$ are optimized jointly to minimize exponential order $\beta_{i}(k)$ according to the attack flow $\tilde{R}(k)$ and system mode $i$. One may notice that the proposed cross-layered defense framework can be applied to other systems apart from PPLSs. Compared with the pure resilient control \cite{res-Italy-2015,res-Feng-2017,res-Sun-2017,multi-Lu-2018}, we bridge the gap between network deployment and system dynamics with the state-space model (\ref{PPLS}) embedded by (\ref{PPLS-cond}). Based on it, we can perform joint optimization to enhance flexibility and resilience.
\end{remark}

\begin{theorem}
Problem (\ref{PPLS-opt}) has an optimal solution under any attack flow $\tilde{R}(k)$.
\end{theorem}

\begin{proof}
One can find that $L_{i}(k)$ is the intermediate variable bridging the gap between $(\tilde{R}(k), W(k))$ and $(K_{i}(k), \beta_{i}(k))$. There are $2^{n}-1$ possible outcomes of $L_{i}(k)$ dependent on $\tilde{R}(k)$ and $W(k)$, especially $L_{i}(k)=\mathrm{diag}(0,0,\ldots,0)$ is excluded according to (\ref{PPLS-opt-source}), which is surely satisfied according to Assumption \ref{assump-one} and Remark \ref{remark-one}.

For any possible outcome of $L_{i}(k)$, Problem (\ref{PPLS-opt}) can be reduced to a convex optimization problem with constraints (\ref{PPLS-opt-LMI-1}) and (\ref{PPLS-opt-LMI-2}), which can be rewritten as follows with the Schur complement.
\begin{equation}\label{PPLS-fea-Schur-1}
\begin{aligned}
(A_{i}+B_{i}K_{i}(k)L_{i}(k))^{\mathrm{T}}P_{i-1}(A_{i}+B_{i}K_{i}(k)L_{i}(k))\\
-\frac{\beta_{i}(k)}{T_{i}}((T_{i}+1)P_{i-1}-P_{i})\preceq 0,
\end{aligned}
\end{equation}

\begin{equation}\label{PPLS-fea-Schur-2}
\begin{aligned}
(A_{i}+B_{i}K_{i}(k)L_{i}(k))^{\mathrm{T}}P_{i}(A_{i}+B_{i}K_{i}(k)L_{i}(k))\\
-\frac{\beta_{i}(k)}{T_{i}}(P_{i-1}+(T_{i}-1)P_{i})\preceq 0,
\end{aligned}
\end{equation}
where matrices $P_{i-1}\succ 0$, $P_{i}\succ 0$ are obtained from Lemma~\ref{PPLS-lemma-alpha}. According to the proof of Corollary 1 and (14) in \cite{PPLS-JFI-2022}, $P_{i-1},P_{i}$ satisfy
\begin{equation}
(A_{i}+B_{i}K_{i})^{\mathrm{T}}P_{i-1}(A_{i}+B_{i}K_{i})-\frac{\alpha_{i}}{T_{i}}((T_{i}+1)P_{i-1}-P_{i})\prec 0,
\end{equation}
which implies $(T_{i}+1)P_{i-1}\succ P_{i}$. Then with $(T_{i}+1)P_{i-1}-P_{i}\succ 0$ and $P_{i-1}+(T_{i}-1)P_{i}\succ 0$, for any $K_{i}(k)$ and $L_{i}(k)$, (\ref{PPLS-fea-Schur-1}) and (\ref{PPLS-fea-Schur-2}) holds strictly with a large enough $\beta_{i}(k)>0$. Hence, Problem (\ref{PPLS-opt}) is strictly feasible for any channels' state $L_{i}(k)$ within the $2^{n}-1$ possible outcomes.

With the strict feasibility and Slater's condition \cite{opt-convex}, Problem (\ref{PPLS-opt}) holds strong duality. According to Lemma \ref{PPLS-lemma-beta}, (\ref{PPLS-opt-LMI-1}) and (\ref{PPLS-opt-LMI-2}), the objective value $\beta_{i}(k)$ satisfies $V_{i}(k+1)\leq \beta_{i}(k)V_{i}(k)$. With the positive definiteness of $V_{i}(k)$, one has $\beta_{i}(k)\geq 0$. According to the strong duality and the lower boundedness of $\beta_{i}(k)$, Problem (\ref{PPLS-opt}) has an optimal solution under any attack flow $\tilde{R}(k)$.
\end{proof}

\subsection{Linearization of online optimization}
With regard to the nonlinear conditions (\ref{PPLS-opt-logic}), (\ref{PPLS-opt-LMI-1}) and (\ref{PPLS-opt-LMI-2}), we have the following two propositions.

\begin{proposition}\label{prop-logic}
The enabling condition (\ref{PPLS-opt-logic}) can be linearized equivalently as the following inequality system:
\begin{equation}\label{opt-bigM}
\resizebox{0.95\hsize}{!}{$
\begin{aligned}
\left(R^{(j)}+\tilde{R}^{(j)}(k)-W^{(j)}(k-1)\right)\tau-S^{(j)}&<M(1-z_{1}^{(j)}(k)),\\
\left(R^{(j)}+\tilde{R}^{(j)}(k)-W^{(j)}(k-1)\right)\tau-S^{(j)}&\geq -M z_{1}^{(j)}(k),\\
W^{(j)}(k)-R^{(j)}-\tilde{R}^{(j)}(k) &\geq -M (1-z_{2}^{(j)}(k)),\\
W^{(j)}(k)-R^{(j)}-\tilde{R}^{(j)}(k) &<M z_{2}^{(j)}(k),\\
l_{i}^{(j)}(k)\leq z_{1}^{(j)}(k),\ l_{i}^{(j)}(k)\leq z_{2}^{(j)}(k),\ l_{i}^{(j)}(k)&\geq z_{1}^{(j)}(k)+z_{2}^{(j)}(k)-1,
\end{aligned}
$}
\end{equation}
where
\begin{equation*}
\resizebox{0.95\hsize}{!}{$
\begin{aligned}
&z_{1}^{(j)}(k),z_{2}^{(j)}(k)\in \{0,1\},\:M=\max \left\{M_{1},M_{2},M_{3},M_{4}\right\},\\
&M_{1}=\max\limits_{j}\big\{\big(R^{(j)}+\bar{R}^{(j)}\big)\tau-S^{(j)}\big\},\:M_{3}=\max\limits_{j}\big\{R^{(j)}+\bar{R}^{(j)}\big\},\\
&M_{2}=\max\limits_{j}\big\{\big(W_{\Sigma}-R^{(j)}\big)\tau+S^{(j)}\big\},\:M_{4}=\max\limits_{j}\big\{W_{\Sigma}-R^{(j)}\big\}.
\end{aligned}
$}
\end{equation*}
Indeed $z_{1}^{(j)}(k)=1$ indicates $\big(R^{(j)}+\tilde{R}^{(j)}(k)-W^{(j)}(k-1)\big)\tau< S^{(j)}$, and $z_{1}^{(j)}(k)=0$ indicates the violation. Similarly, $z_{2}^{(j)}(k)$ indicates whether $W^{(j)}(k)\geq R^{(j)}+\tilde{R}^{(j)}(k)$ or not.\hfill $\Box$
\end{proposition}

\begin{proof}
Referring to Chapter 9.2 of \cite{opt-logic}, with big-M method,
\begin{equation}
z_{1}^{(j)}(k)=\fontsize{9}{11}\selectfont{\begin{cases}
1, & \big(R^{(j)}+\tilde{R}_{i}^{(j)}(k)-W_{i}^{(j)}(k-1)\big)\tau< S^{(j)},\\
0, & \text{otherwise}.
\end{cases}}
\end{equation}
can be linearized equivalently as
\begin{equation}
\resizebox{0.95\hsize}{!}{$
\begin{aligned}
\big(R^{(j)}+\tilde{R}^{(j)}(k)-W^{(j)}(k-1)\big)\tau-S^{(j)}&<M(1-z_{1}^{(j)}(k)),\\
\big(R^{(j)}+\tilde{R}^{(j)}(k)-W^{(j)}(k-1)\big)\tau-S^{(j)}&\geq -M z_{1}^{(j)}(k),
\end{aligned}
$}
\end{equation}
since
\begin{equation}
\resizebox{0.91\hsize}{!}{$
\begin{aligned}
M\geq&\max\{M_{1},M_{2}\}\\
\geq&\sup\big|\big(R^{(j)}+\tilde{R}^{(j)}(k)-W^{(j)}(k-1)\big)\tau-S^{(j)}\big|.
\end{aligned}
$}
\end{equation}
Similarly,
\begin{equation}
z_{2}^{(j)}(k)=\begin{cases}
1, & W^{(j)}(k)\geq R^{(j)}+\tilde{R}^{(j)}(k),\\
0, & \text{otherwise}.
\end{cases}
\end{equation}
can be linearized equivalently as
\begin{equation}
\begin{aligned}
W^{(j)}(k)-R^{(j)}-\tilde{R}^{(j)}(k) &\geq -M (1-z_{2}^{(j)}(k)),\\
W^{(j)}(k)-R^{(j)}-\tilde{R}^{(j)}(k) &< M z_{2}^{(j)}(k),
\end{aligned}
\end{equation}
since
\begin{equation}
\resizebox{0.66\hsize}{!}{$
\begin{aligned}
M\geq&\max\{M_{3},M_{4}\}\\
\geq&\sup\big|W^{(j)}(k)-R^{(j)}-\tilde{R}^{(j)}(k)\big|.
\end{aligned}
$}
\end{equation}

Furthermore, the logical condition
\begin{equation}
l_{i}^{(j)}(k)=\begin{cases}
1, & z_{1}^{(j)}(k)=1\text{ and }z_{2}^{(j)}(k)=1,\\
0, & \text{otherwise}.
\end{cases}
\end{equation}
is equivalent to
\begin{equation}
\resizebox{0.98\hsize}{!}{$
l_{i}^{(j)}(k)\leq z_{1}^{(j)}(k),\ l_{i}^{(j)}(k)\leq z_{2}^{(j)}(k),\ l_{i}^{(j)}(k)\geq z_{1}^{(j)}(k)+z_{2}^{(j)}(k)-1.
$}
\end{equation}
Noting that $z_{1}^{(j)}(k)$ indicates whether $\big(R^{(j)}+\tilde{R}^{(j)}(k)-W^{(j)}(k-1)\big)\tau< S^{(j)}$, $z_{2}^{(j)}(k)$ indicates whether $W^{(j)}(k)\geq R^{(j)}+\tilde{R}^{(j)}(k)$, the equivalence of linearization in Proposition \ref{prop-logic} can be proved.
\end{proof}

\begin{proposition}\label{prop-product}
Assume that the controller gain $K_{i}(k)=[k_{i,pq}(k)]_{n\times n}$ is bounded as $-\bar{k}\leq k_{i,pq}(k)\leq \bar{k}$, then $K_{i}(k)L_{i}(k)$ in (\ref{PPLS-opt-LMI-1}) and (\ref{PPLS-opt-LMI-2}) can be linearized equivalently as $Q_{i}(k)-\bar{k}\mathbf{1}_{n\times n}L_{i}(k)$
and
\begin{equation}\label{PPLS-opt-product-linear}
\begin{aligned}
2\bar{k}\mathbf{1}_{n \times n}\geq Q_{i}(k)\geq 0,\ Q_{i}(k)-2\bar{k}\mathbf{1}_{n\times n}L_{i}(k)&\leq 0,\\ Q_{i}(k)-K_{i}(k)-\bar{k}\mathbf{1}_{n \times n} &\leq 0,\\
Q_{i}(k)-K_{i}(k)+\bar{k}\mathbf{1}_{n\times n}(I-2L_{i}(k)) &\geq 0,
\end{aligned}
\end{equation}
where (\ref{PPLS-opt-product-linear}) forces $Q_{i}(k)-\bar{k}\mathbf{1}_{n\times n}L_{i}(k)$ to take the value of $K_{i}(k)L_{i}(k)$.\hfill $\Box$
\end{proposition}

\begin{proof}
Let $k_{i}\in [-\bar{k},\bar{k}]$, $l_{i}\in\{0,1\}$ and $\hat{k}_{i}\triangleq k_{i}+\bar{k}\in [0,2\bar{k}]$. Referring to (11)--(14) in \cite{opt-product}, $\hat{k}_{i}l_{i}$ can be replaced equivalently with $q_{i}$ along with
\begin{equation}\label{PPLS-opt-proof-product-hat-k}
\begin{aligned}
2\bar{k}\geq q_{i} &\geq 0,\\
q_{i}-2\bar{k}l_{i}&\leq 0,\\
q_{i}-\hat{k}_{i} &\leq 0,\\
q_{i}-\hat{k}_{i}+2\bar{k}(1-l_{i}) &\geq 0,
\end{aligned}
\end{equation}
which forces $q_{i}$ to take the value of $k_{i}l_{i}$. With $\hat{k}_{i}= k_{i}+\bar{k}$, one has that $k_{i}l_{i}$ can be replaced equivalently with $q_{i}-\bar{k}l_{i}$ along with
\begin{equation}\label{PPLS-opt-proof-product-k}
\begin{aligned}
2\bar{k}\geq q_{i} &\geq 0,\\
q_{i}-2\bar{k}l_{i}&\leq 0,\\
q_{i}-k_{i}-\bar{k} &\leq 0,\\
q_{i}-k_{i}+\bar{k}-2\bar{k}l_{i} &\geq 0.
\end{aligned}
\end{equation}
Then, Proposition \ref{prop-product} can be obtained as an extension of (\ref{PPLS-opt-proof-product-k}) in form of matrices.
\end{proof}

\begin{remark}
With Propositions \ref{prop-logic} and \ref{prop-product}, the nonlinear conditions (\ref{PPLS-opt-logic}), (\ref{PPLS-opt-LMI-1}) and (\ref{PPLS-opt-LMI-2}) are transformed into linear inequalities, which can be solved by common numerical softwares.
\end{remark}

\section{Stability analysis}\label{section-stability}
As the focus of the defense strategy, the single-level optimization problem (\ref{PPLS-opt}) is solved online to minimize exponential order $\beta_{i}(k)$ according to attack flow $\tilde{R}(k)$ and system mode $i$. The bi-level mixed-integer semidefinite programming (\ref{PPLS-opt-bi}) is solved offline to establish the worst case (supremum $\bar{\beta}_{i}$) of the optimal result $\beta_{i}^{*}(k)$ in Problem (\ref{PPLS-opt}). Then $\bar{\beta}_{i}$ is synthesized with the exponential order $\alpha_{i}$ in the absence of attacks to analyze the stability of the closed-loop system under the defense strategy.

\subsection{Worst case of online optimization}
Observing the online optimization problem (\ref{PPLS-opt}), one can find that the optimal result $\beta_{i}^{*}(k)$ is dependent on the real-time attack flow $\tilde{R}(k)$. By solving the following bi-level problem offline, we can establish the worst case (supremum $\bar{\beta}_{i}$) of $\beta_{i}^{*}(k)$ and the corresponding attack flow.

\begin{subequations}\label{PPLS-opt-bi}
\begin{align}
&\bar{\beta}_{i}=\max\limits_{\tilde{R}}\ \min\limits_{W,L_{i},K_{i},\beta_{i}}\ \beta_{i}\\
\text{s.t.}\quad &l_{i}^{(j)}=\fontsize{9}{11}\selectfont{\begin{cases}
1, & \big(R^{(j)}+\tilde{R}^{(j)}\big)\tau< S^{(j)}\\
  & \text{and}\ W^{(j)}\geq R^{(j)}+\tilde{R}^{(j)},\\
0, & \text{otherwise}.
\end{cases}}\label{PPLS-opt-bi-logic}\\
&\sum_{j=1}^{n}W^{(j)}\leq W_{\Sigma},\ W^{(j)}\geq 0,\ \sum_{j=1}^{n} l_{i}^{(j)}\geq 1,\\
&\sum_{j=1}^{n}\tilde{R}^{(j)}\leq \tilde{R}_{\Sigma},\ 0\leq\tilde{R}^{(j)}\leq \bar{R}^{(j)},\label{PPLS-opt-bi-attack}\\
&\begin{bmatrix}
-\frac{T_{i}+1}{T_{i}}\beta_{i}P_{i-1}+\frac{\beta_{i}}{T_{i}}P_{i} & *\\
A_{i}+B_{i}K_{i}L_{i} & -P_{i-1}^{-1}
\end{bmatrix}\preceq 0,\\
&\begin{bmatrix}
-\frac{T_{i}-1}{T_{i}}\beta_{i}P_{i}-\frac{\beta_{i}}{T_{i}}P_{i-1} & *\\
A_{i}+B_{i}K_{i}L_{i} & -P_{i}^{-1}
\end{bmatrix}\preceq 0,
\end{align}
\end{subequations}
where $W^{(j)}(k-1)=0$ is given in (\ref{PPLS-opt-bi-logic}) to investigate the worst case. Note that Problem (\ref{PPLS-opt-bi}), where attack flow $\tilde{R}$ serves as a variable, is solved offline without dependence on real-time attack flow.
\begin{remark}
One may notice that the inner-level minimization problem of Problem (\ref{PPLS-opt-bi}) is non-convex, Problem (\ref{PPLS-opt-bi}) cannot be reduced to a single-level maximization problem with Karush-Kuhn-Tucker conditions or strong dual method \cite{opt-convex}. It motivates us to propose a customized solution algorithm: smart enumeration algorithm (SEA).
\end{remark}

\subsection{Smart enumeration algorithm}
SEA is proposed to solve Problem (\ref{PPLS-opt-bi}) exactly. To clarify further discussion, some notations are given as follows.
\begin{itemize}
  \item $\mathcal{L}$, the main diagonal vector of $L_{i}$, denotes the channels' state, with $\mathcal{L}(\mathbf{1})$, $\mathcal{L}(\mathbf{0})$ denoting the enabled channels and jammed channels, respectively. For example, one has $\mathcal{L}(\mathbf{1})=\{2,3\}$ and $\mathcal{L}(\mathbf{0})=\{1,4\}$, if $\mathcal{L}=[0\:1\:1\:0]$.
  \item $\mathcal{S}_{all}$ denotes all the states of $\mathcal{L}$. For system (\ref{PPLS}) with $n$ channels, there are $2^{n}$ elements in $\mathcal{S}_{all}$.
  \item $\mathcal{S}_{force}$ denotes the set of the channels' states that the attack flow $\tilde{R}$ can induce forcibly by violating the delay condition $\big(R^{(j)}+\tilde{R}^{(j)}\big)\tau< S^{(j)}$ in (\ref{PPLS-opt-bi-logic}), which is solely dependent on $\tilde{R}^{(j)}$. For a set $\{i_{1},i_{2},\ldots,i_{p}\}$ where $i_{h} \in \{1,2,\ldots,n\}$ for all $h \in \{1,2,\ldots,p\}$, if it holds
      \begin{equation}\label{S-force}
      \resizebox{0.9\hsize}{!}{$
        \sum\limits_{h=1}^{p} \Big(\frac{S^{\left(i_{h}\right)}}{\tau}-R^{\left(i_{h}\right)}\Big)\leq \tilde{R}_{\Sigma},\ \frac{S^{\left(i_{h}\right)}}{\tau}-R^{\left(i_{h}\right)}\leq \bar{R}^{(i_{h})},
      $}
      \end{equation}
      then
      \begin{equation}\label{L-force}
        \mathcal{L}_{f}=
        \begin{bNiceMatrix}[last-row,nullify-dots]
        \text{?} & \Cdots & \text{?} & 0              & \text{?} & \Cdots & \text{?} & 0              & \text{?} & \Cdots & \text{?}\\
        & & & i_{1}\text{th} & \Cdots & & & i_{p}\text{th} & & &
        \end{bNiceMatrix}
      \end{equation}
      is an element in $\mathcal{S}_{force}$ that is obtained by verifying (\ref{S-force}) for all the elements in $\mathcal{S}_{all}$, respectively. In (\ref{L-force}), the $i_{1},i_{2},\ldots,i_{p}$th channels are jammed forcibly by violating the delay condition, and `?' denotes the channel state dependent on the following bandwidth allocation.
  \item $\mathcal{S}_{safe}\mid\mathcal{L}_{f}$ denotes the set of channels' states that the defender can definitely reach by allocating bandwidth under $\mathcal{L}_{f}$, that is, when the attack flow $\tilde{R}$ jams the channels $\mathcal{L}_{f}(\mathbf{0})$ forcibly by violating the delay condition. For a set $\{j_{1},j_{2},\ldots,j_{q}\}$ which satisfies $\{j_{1},j_{2},\ldots,j_{q}\} \cap \{i_{1},i_{2},\ldots,i_{p}\}=\emptyset$, if it holds
      \begin{equation}\label{S-safe}
      \resizebox{0.83\hsize}{!}{$
      \begin{aligned}
      W_{\Sigma}\geq\min\bigg\{&\tilde{R}_{\Sigma}-\sum_{j\in\mathcal{L}_{f}(\mathbf{0})}\Big(\frac{S^{(j)}}{\tau}-R^{(j)}\Big)+\sum_{h=1}^{q}R^{(j_{h})},\\
      &\sum_{h=1}^{q} \left(R^{(j_{h})}+\bar{R}^{(j_{h})}\right)\bigg\},
      \end{aligned}
      $}
      \end{equation}
      then
      \begin{equation}\label{L-safe}
      \resizebox{0.8\hsize}{!}{$
        \mathcal{L}_{s}=
        \begin{bNiceMatrix}[last-row,nullify-dots]
        0 & \Cdots & 0 & 1              & 0      & \Cdots & 0 & 1              & 0 & \Cdots & 0\\
          &        &   & j_{1}\text{th} & \Cdots &        &   & j_{q}\text{th} &   &        &
        \end{bNiceMatrix}
      $}
      \end{equation}
      is an element in $\mathcal{S}_{safe}\mid\mathcal{L}_{f}$ that is obtained by verifying (\ref{S-safe}) for all possible channels' states under $\mathcal{L}_{f}$. In (\ref{L-safe}), the $j_{1},j_{2},\ldots,j_{q}$th entries are $1$, the others are $0$.
  \item $\mathcal{S}_{opt}$ denotes the subset of $\mathcal{S}_{all}$, in which the elements are the potential optimal channels' states of Problem (\ref{PPLS-opt-bi}).
\end{itemize}

Given a subsystem $(A_{i},B_{i})$, we can obtain the corresponding exact optimal objective value $\bar{\beta}_{i}$ of Problem (\ref{PPLS-opt-bi}) by executing SEA.
\begin{breakablealgorithm}\label{algorithm}
\renewcommand{\thealgorithm}{}
\caption{}
\begin{enumerate}[leftmargin=11mm]\renewcommand{\labelenumi}{\textbf{Step \theenumi}}
  \item Initialize $\mathcal{S}_{opt}=\mathcal{S}_{all}$.
  \item Assign each $\mathcal{L} \in \mathcal{S}_{all}$ to the main diagonal of $L_{i}$ one by one, and solve (\ref{PPLS-opt}) to obtain the corresponding optimal exponential orders $\beta_{i}^{*}\mid\mathcal{L}$, respectively.
  \item Determine the set $\mathcal{S}_{force}$ with (\ref{S-force}) and (\ref{L-force}).
  \item For all $\mathcal{L}_{f} \in \mathcal{S}_{force}$, set the undetermined entries `?' as $1$ to obtain $\beta_{i}^{*}\mid \mathcal{L}_{f}$ according to the result of \textbf{Step~2}, then update
      \begin{equation}\label{tilde-beta}
      \tilde{\beta}_{i}^{*}\triangleq\max\limits_{\mathcal{L}_{f}\in \mathcal{S}_{force}}\left\{\beta_{i}^{*}\mid\mathcal{L}_{f}\right\}.
      \end{equation}
  \item Given a $\mathcal{L}_{f} \in \mathcal{S}_{force}$, determine the set $\mathcal{S}_{safe}\mid\mathcal{L}_{f}$ with (\ref{S-safe}) and (\ref{L-safe}), then find $\hat{\beta}_{i}^{*}\mid\mathcal{L}_{f}\triangleq\min\limits_{\mathcal{L}_{s}\in\mathcal{S}_{safe}\mid\mathcal{L}_{f}} \left\{\beta_{i}^{*}\mid\mathcal{L}_{s}\right\}$ for $\mathcal{S}_{safe}\mid\mathcal{L}_{f}$. Update
      \begin{equation}\label{S-all-Lf}
      \mathcal{S}_{all}\mid \mathcal{L}_{f}\triangleq \{\mathcal{L}\in\mathcal{S}_{all}:\mathcal{L}(\mathbf{1})\cap\mathcal{L}_{f}(\mathbf{0})=\emptyset\},
      \end{equation}
      \begin{equation}\label{S-opt-Lf}
      \begin{aligned}
        &\mathcal{S}_{opt}\mid\mathcal{L}_{f}\\
        \triangleq&\big\{\mathcal{L}_{o} \in \mathcal{S}_{all}\mid \mathcal{L}_{f}:\tilde{\beta}_{i}^{*} \leq \beta_{i}^{*}\mid\mathcal{L}_{o} \leq \hat{\beta}_{i}^{*}\mid\mathcal{L}_{f}\big\}.
      \end{aligned}
      \end{equation}
  \item For the given $\mathcal{L}_{f}$, verify the following condition for all the $\mathcal{L}_{o}\in\mathcal{S}_{opt}\mid\mathcal{L}_{f}$ one by one in an ascending order of $\beta_{i}^{*}\mid\mathcal{L}_{o}$: For a $\mathcal{L}_{o} \in \mathcal{S}_{opt}\mid\mathcal{L}_{f}$, define
      \begin{equation}\label{opt-test-S}
      \resizebox{0.88\hsize}{!}{$
      \begin{aligned}
        &\mathcal{S}_{\mathcal{L}_{o}}\mid\mathcal{L}_{f}\\
        \triangleq&\left\{\mathcal{L}^{\prime}\in (\mathcal{S}_{all}\mid \mathcal{L}_{f})\backslash \mathcal{L}_{o}:\beta_{i}^{*}\mid\mathcal{L}_{f}\leq\beta_{i}^{*}\mid\mathcal{L}^{\prime}< \beta_{i}^{*}\mid\mathcal{L}_{o}\right\}.
      \end{aligned}
      $}
      \end{equation}
      If under (\ref{PPLS-opt-bi-attack}) and
      \begin{equation}\label{opt-test-if}
      \resizebox{0.88\hsize}{!}{$
      W_{\Sigma}<\sum\limits_{p\in \mathcal{L}_{o}(\mathbf{1})}\big(R^{(p)}+\tilde{R}^{(p)}\big)\leq\tilde{R}_{\Sigma}-\sum\limits_{j\in\mathcal{L}_{f}(\mathbf{0})}\Big(\frac{S^{(j)}}{\tau}-R^{(j)}\Big),
      $}
      \end{equation}
      there does not exist an attack flow $\tilde{R}$ such that
      \begin{equation}\label{opt-test}
      \forall \mathcal{L}^{\prime}\in\mathcal{S}_{\mathcal{L}_{o}}\mid\mathcal{L}_{f},\ \sum_{q\in \mathcal{L}^{\prime}(\mathbf{1})}\big(R^{(q)}+\tilde{R}^{(q)}\big)> W_{\Sigma},
      \end{equation}
      stop verifying the remaining elements in $\mathcal{S}_{opt}\mid\mathcal{L}_{f}$, and set
      \begin{equation}\label{opt-test-bar}
        \bar{\beta}_{i}\mid\mathcal{L}_{f}= \beta_{i}^{*}\mid\mathcal{L}_{o},
      \end{equation}
      where $\bar{\beta}_{i}\mid\mathcal{L}_{f}$ denotes the maximum exponential order that the attack flow $\tilde{R}$ can induce under $\mathcal{L}_{f}$, that is, such that
      \begin{equation}
        \forall j\in \mathcal{L}_{f}(\mathbf{0}),\ \tilde{R}^{(j)}\geq \frac{S^{(j)}}{\tau}-R^{(j)}.
      \end{equation}
  \item Repeat \textbf{Steps 5, 6} for each $\mathcal{L}_{f}\in \mathcal{S}_{force}$ to obtain $\bar{\beta}_{i}\mid\mathcal{L}_{f}$, then we have the optimal objective value of (\ref{PPLS-opt-bi})
      \begin{equation}\label{opt-syn}
        \bar{\beta}_{i}=\max\limits_{\mathcal{L}_{f}\in \mathcal{S}_{force}}\left\{\bar{\beta}_{i}\mid\mathcal{L}_{f}\right\}.
      \end{equation}
\end{enumerate}
\end{breakablealgorithm}

Before proving the optimality of the algorithm, we give the following lemma.

\begin{lemma}\label{lemma-test}
$\bar{\beta}_{i}\mid \mathcal{L}_{f}$ in \textbf{Step 6} is optimal for Problem (\ref{PPLS-opt-bi}) under a given $\mathcal{L}_{f}$.\hfill $\Box$
\end{lemma}
\begin{proof}
Define
\begin{equation}\label{Lo-}
\mathcal{L} _{o}^{-}\triangleq\mathop{\arg\max}\limits_{\mathcal{L}\in\mathcal{S}_{opt}\mid\mathcal{L}_{f}}\{\beta_{i}^{*}\mid\mathcal{L}:\beta_{i}^{*}\mid\mathcal{L}<\beta_{i}^{*}\mid\mathcal{L}_{o}\}.
\end{equation}
Note that $\mathcal{L}_{o}^{-}$ is verified before $\mathcal{L}_{o}$ in \textbf{Step 6}, hence $\mathcal{L}_{o}^{-}$ satisfies (\ref{opt-test}). That is, there exists a attack flow $\tilde{R}$ such that (\ref{opt-test}) for all the $\mathcal{L}^{\prime}\in\mathcal{S}_{\mathcal{L}_{o}^{-}}\mid\mathcal{L}_{f}$ under (\ref{PPLS-opt-bi-attack}) and
\begin{equation}
\resizebox{0.88\hsize}{!}{$
W_{\Sigma}<\sum\limits_{p\in \mathcal{L}_{o}^{-}(\mathbf{1})}\big(R^{(p)}+\tilde{R}^{(p)}\big)\leq\tilde{R}_{\Sigma}-\sum\limits_{j\in\mathcal{L}_{f}(\mathbf{0})}\Big(\frac{S^{(j)}}{\tau}-R^{(j)}\Big).
$}
\end{equation}
Under such an attack flow $\tilde{R}$, the defender will reach a channels' state with an exponential order larger than $\beta_{i}^{*}\mid \mathcal{L}_{o}^{-}$ after bandwidth allocation. It implies that $\bar{\beta}_{i}\mid \mathcal{L}_{f}>\beta_{i}^{*}\mid \mathcal{L}_{o}^{-}$. Furthermore, it holds infeasibility of (\ref{opt-test}) under (\ref{PPLS-opt-bi-attack}) and (\ref{opt-test-if}) for $\mathcal{L}_{o}$. It follows that $\bar{\beta}_{i}\mid \mathcal{L}_{f}\leq\beta_{i}^{*}\mid \mathcal{L}_{o}$. Since $\beta_{i}^{*}\mid \mathcal{L}_{o}^{-}<\bar{\beta}_{i}\mid \mathcal{L}_{f}\leq\beta_{i}^{*}\mid \mathcal{L}_{o}$, recalling the definition (\ref{Lo-}) of $\mathcal{L}_{o}^{-}$, we can obtain $\bar{\beta}_{i}\mid \mathcal{L}_{f}=\beta_{i}^{*}\mid \mathcal{L}_{o}$.
\end{proof}

\begin{theorem}\label{theo-algorithm}
The result $\bar{\beta}_{i}$ from SEA is the exact optimal value of Problem (\ref{PPLS-opt-bi}).
\end{theorem}
\begin{proof}
SEA partitions the solution space of (\ref{PPLS-opt-bi}) into finite parts with respect to all the possible states of $\mathcal{L}_{f}$. Then according to Lemma \ref{lemma-test}, the local optimum $\bar{\beta}_{i}\mid \mathcal{L}_{f}$ of each part can be obtained by \textbf{Step 6}. The global optimum $\bar{\beta}_{i}$ can be obtained by synthesizing all the local optimums with (\ref{opt-syn}).
\end{proof}

\begin{remark}
In (\ref{S-opt-Lf}), the channels' states that cannot be globally optimal for (\ref{PPLS-opt-bi}) are filtered out from $\mathcal{S}_{opt}\mid \mathcal{L}_{f}$. For a given $\mathcal{L}_{f}$, if the local optimum is with an exponential order less than $\tilde{\beta}_{i}^{*}$ in (\ref{tilde-beta}), the $\bar{\beta}_{i}\mid \mathcal{L}_{f}$ obtained from \textbf{Step 6} would be larger than the real local optimum. However, it will not affect the exactness of the global optimum $\bar{\beta}_{i}$ obtained by SEA, since the obtained $\bar{\beta}_{i}\mid \mathcal{L}_{f}$ satisfies $\bar{\beta}_{i}\mid \mathcal{L}_{f}=\tilde{\beta}_{i}^{*}\leq \bar{\beta}_{i}$ in this case.
\end{remark}

\begin{remark}
The computation workload of SEA is rather heavy, but Problem (\ref{PPLS-opt-bi}) is solved with SEA only once offline for each subsystem, instead of being solved online at each time $k$. Moreover, to improve the efficiency of \textbf{Step 6}, $\tilde{\beta}_{i}^{*}$, $\hat{\beta}_{i}^{*}\mid\mathcal{L}_{f}$ are determined in \textbf{Steps 4, 5} to narrow $\mathcal{S}_{opt}\mid \mathcal{L}_{f}$, the set of potential optimal channels' states under $\mathcal{L}_{f}$.
\end{remark}

\subsection{Stability analysis}
As discussed in the start of Section \ref{section-stability}, $\bar{\beta}_{i}$ obtained from Problem (\ref{PPLS-opt-bi}) is the supremum of the exponential order $\beta_{i}^{*}(k)$ under online defense (\ref{PPLS-opt}) against DoS attacks. The exponential order $\alpha_{i}$ in the absence of attacks is established in Lemma \ref{PPLS-lemma-alpha}. With Assumption \ref{PPLS-assump-intensity}, $\alpha_{i}$ and $\bar{\beta}_{i}$, Theorem \ref{PPLS-theo-stability} is proposed to analyze the stability of system (\ref{PPLS}) under online defense (\ref{PPLS-opt}).

\begin{theorem}\label{PPLS-theo-stability}
Consider a multi-channel PPLS (\ref{PPLS}) under online defense (\ref{PPLS-opt}) against DoS attacks. Given $0\leq\delta_{i}\leq 1$ from Assumption \ref{PPLS-assump-intensity}. Establish $\alpha_{i}$ and $P_{i}$ such that (\ref{PPLS-lemma-alpha-LMI-1}) and (\ref{PPLS-lemma-alpha-LMI-2}) in Lemma \ref{PPLS-lemma-alpha}, then obtain $\bar{\beta}_{i}>0$ from Problem (\ref{PPLS-opt-bi}) with $P_{i}$, $i=1,2,\ldots,s$. If there exists $0<\chi<1$ such that
\begin{equation}\label{PPLS-cond-sta-order}
\prod_{i=1}^{s}\alpha_{i}^{(1-\delta_{i})T_{i}}{\bar{\beta}_i}^{\delta_{i}T_{i}}\leq \chi^{2T},
\end{equation}
then system (\ref{PPLS}) is $\chi$-exponentially stable.
\end{theorem}

\begin{proof}
Notice that
\begin{equation}\label{proof-sta-difference}
V_{i}(k+1)\leq\begin{cases}
\alpha_{i}V_{i}(k), & \text{DoS attacks stop},\\
\beta_{i}^{*}(k)V_{i}(k)\leq \bar{\beta}_{i}V_{i}(k), & \text{DoS attacks occur}.
\end{cases}
\end{equation}
With Assumption \ref{PPLS-assump-intensity}, the evolution of $V(k)$ over the periods follows that
\begin{equation}
\begin{aligned}
V(\ell T)&=V_{1}(\ell T)\\
&\leq \alpha_{s}^{(1-\delta_{s})T_{s}}{\bar{\beta}_s}^{\delta_{s}T_{s}} V_{s}((\ell-1) T+k_{s-1})\\
&\leq \left(\prod_{i=1}^{s}\alpha_{i}^{(1-\delta_{i})T_{i}}\bar{\beta}_{i}^{\delta_{i}T_{i}}\right)V((\ell-1) T),
\end{aligned}
\end{equation}
where $\delta_{i}=\frac{\tilde{T}_{i}}{T_{i}}$. Then with (\ref{PPLS-cond-sta-order}), one has
\begin{equation}
\begin{aligned}
V(\ell T)&\leq \left(\prod_{i=1}^{s}\alpha_{i}^{(1-\delta_{i})T_{i}}\bar{\beta}_{i}^{\delta_{i}T_{i}}\right)^{\ell}V_{1}(0)\\
&\leq\chi^{2\ell T}V(0).
\end{aligned}
\end{equation}
Since $V(0)=x^{\mathrm{T}}(0)P_{s}x(0)$ and $V(\ell T)=x^{\mathrm{T}}(\ell T)P_{s}x(\ell T)$, one has
\begin{equation}\label{proof-lTp-0}
\|x(\ell T)\|\leq \sqrt{\frac{\bar{\lambda}(P_{s})}{\underline{\lambda}(P_{s})}}\chi^{\ell T}\|x(0)\|.
\end{equation}

Furthermore, for $k \in \{\ell T+k_{i-1},\ldots,\ell T+k_{i}-1\}$, notice that $P_{i}(k)$ varies linearly with $k$ in (\ref{PPLS-Lya}), one has
\begin{equation}
\|x(k+1)\|\leq \sqrt{\bar{\beta}_{i} \frac{\max\left\{\bar{\lambda}(P_{i}),\bar{\lambda}(P_{i-1})\right\}}{\min\left\{\underline{\lambda}(P_{i}),\underline{\lambda}(P_{i-1})\right\}}}\|x(k)\|.
\end{equation}
Then with (\ref{proof-lTp-0}), we can obtain
\begin{equation}
\begin{aligned}
\|x(k)\|&\leq \left(\prod_{i=1}^{s}\theta_{i}\right) \|x(\ell T)\|\\
&\leq \left(\prod_{i=1}^{s}\theta_{i}\right)\sqrt{\frac{\bar{\lambda}(P_{s})}{\underline{\lambda}(P_{s})}}\chi^{\ell T}\|x(0)\|\\
&=\left(\prod_{i=1}^{s}\theta_{i}\right)\sqrt{\frac{\bar{\lambda}(P_{s})}{\underline{\lambda}(P_{s})}}\frac{1}{\chi^{T}} \chi^{(\ell+1) T}\|x(0)\|\\
&\leq \left(\prod_{i=1}^{s}\theta_{i}\right)\sqrt{\frac{\bar{\lambda}(P_{s})}{\underline{\lambda}(P_{s})}}\frac{1}{\chi^{T}} \chi^{k}\|x(0)\|\\
&=c\chi^{k}\|x(0)\|,
\end{aligned}
\end{equation}
where
\begin{equation*}
\begin{aligned}
c&=\left(\prod_{i=1}^{s}\theta_{i}\right)\sqrt{\frac{\bar{\lambda}(P_{s})}{\underline{\lambda}(P_{s})}}\frac{1}{\chi^{T}},\\
\theta_{i}&=\max\left\{1, \left(\bar{\beta}_{i} \frac{\max\left\{\bar{\lambda}(P_{i}),\bar{\lambda}(P_{i-1})\right\}}{\min\left\{\underline{\lambda}(P_{i}),\underline{\lambda}(P_{i-1})\right\}}\right)^{\frac{T_{i}}{2}}\right\}.
\end{aligned}
\end{equation*}
Hence system (\ref{PPLS}) is $\chi$-exponentially stable.
\end{proof}

\section{Numerical examples}
In this section, with numerical examples, the effectiveness of the online cross-layered defense strategy, and the adaptivity of the strategy to attack flow and system dynamics are illustrated. And the resilience enhanced by the cross-layered framework is highlighted by comparison. Numerical computation is conducted by using YALMIP \cite{YALMIP} and MOSEK \cite{MOSEK}.

Consider a multi-channel PPLS $\Sigma$ in form of (\ref{PPLS}), which is composed of three subsystems:
\begin{equation*}
\resizebox{0.98\hsize}{!}{$
\begin{aligned}
A_{1}&=\begin{bmatrix}
 1.0526 & -0.0066 & -0.2211 & 0.2816\\
 -0.0500 & 1.1000 & -0.1500 & 0.2000\\
 -0.0395 & -0.0013 & 0.9658 & 0.1763\\
 -0.0224 & -0.0066 & -0.1461 & 1.2816
\end{bmatrix},\;
B_{1}=\begin{bmatrix}
 0.8000 & 0.4000\\
 0.4000 & 0.8000\\
 0.4000 & 0.4000\\
 0.6000 & 0.4000
\end{bmatrix},\\
A_{2}&=\begin{bmatrix}
 0.4434 & -1.6224 & 0.2039 & 1.8211\\
 -0.2461 & -0.4066 & 0.2776 & 1.1474\\
 -0.5395 & -1.3342 & 1.2237 & 1.5263\\
 -0.1592 & -1.2513 & 0.2355 & 1.9895
\end{bmatrix},\;
B_{2}=\begin{bmatrix}
 0.8000 & 0.4000\\
 0.4000 & 0.4000\\
 0.4000 & 0.8000\\
 0.6000 & 0.4000
\end{bmatrix},\\
A_{3}&=\begin{bmatrix}
 -0.5947 & 1.2421 & -0.2632 & 0.3079\\
 -1.3671 & 2.0132 & -0.1697 & 0.2618\\
 -1.7658 & 1.9737 & 0.3395 & 0.2263\\
 -1.4553 & 1.6579 & -0.5868 & 1.1921
\end{bmatrix},\;
B_{3}=\begin{bmatrix}
 0.4000 & 0.4000\\
 0.6000 & 0.4000\\
 0.8000 & 0.4000\\
 0.4000 & 0.8000
\end{bmatrix}.
\end{aligned}
$}
\end{equation*}

The dwell-times of subsystems are $T_{1}=4$, $T_{2}=5$, $T_{3}=6$. One can observe that all the subsystems are unstable and the 1st subsystem is non-stabilizable. Fig.~\ref{PPLS-open} presents the open-loop state response of system $\Sigma$ initialized with $x(0)=[2\:3.2\:1.3\:3]^{\mathrm{T}}$. As can be observed from the plot, system $\Sigma$ exhibits unstable open-loop behavior.

\vspace{-0.4cm}
\begin{figure}[htb]
\centering
\includegraphics[width=.43\textwidth]{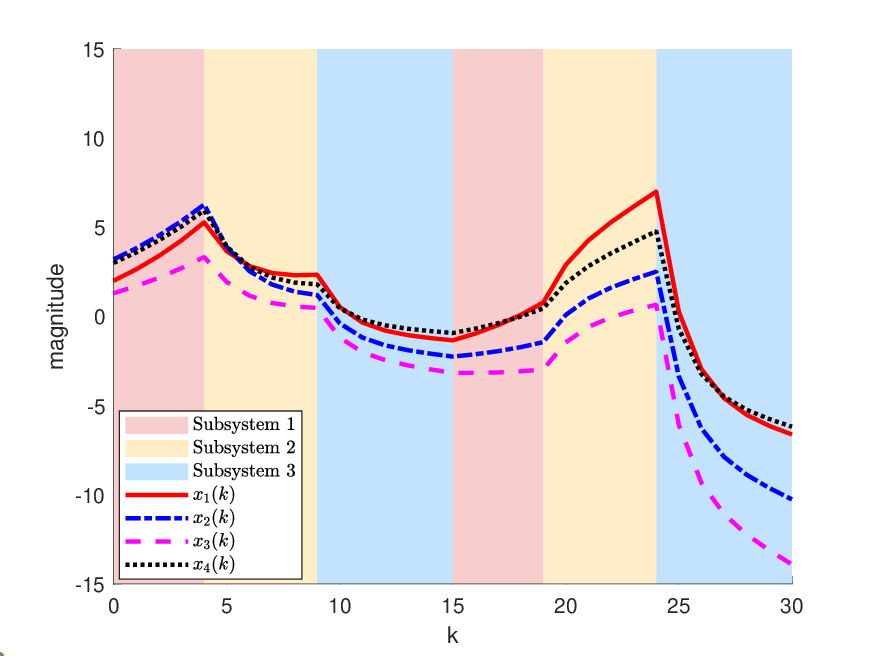}
\caption{Open-loop state response of system $\Sigma$}
\label{PPLS-open}
\end{figure}

\vspace{-0.4cm}
For the network parameters, given the buffer size $S=[10,10,10,10]$, normal flow $R=[5,5,5,5]$, total available bandwidth $W_{\Sigma}=20$ and allocation delay $\tau=0.5$. For the attack parameters, given the total available attack flow $\tilde{R}_{\Sigma}=20$, the upper bounds of attack flow $\bar{R}=[15,15,15,15]$, and the upper bounds of attack duration over a dwell-time $\tilde{T}_{1}=\tilde{T}_{2}=\tilde{T}_{3}=2$, which implies $\delta_{1}=\frac{1}{2}, \delta_{2}=\frac{2}{5}, \delta_{3}=\frac{1}{3}$. Furthermore, assume that the controller gain $K_{i}(k)=[k_{i,pq}(k)]_{n\times n}$ is bounded as $-100\leq k_{i,pq}(k)\leq 100$.

With Lemma \ref{PPLS-lemma-alpha}, given $\alpha_{1}=1.3$, $\alpha_{2}=0.4$, $\alpha_{3}=0.3$, we can obtain the Lyapunov matrices
\begin{equation*}
\resizebox{0.68\hsize}{!}{$
\begin{aligned}
P_{1}&=\begin{bmatrix}
 1.8998 & -0.8123 & -0.2081 & -0.6923\\
 -0.8123 & 5.4990 & -0.6873 & -1.8391\\
 -0.2081 & -0.6873 & 1.0702 & 0.1122\\
 -0.6923 & -1.8391 & 0.1122 & 4.9075
\end{bmatrix},\\
P_{2}&=\begin{bmatrix}
 2.9788 & -1.0600 & -0.3289 & -0.7628\\
 -1.0600 & 7.9116 & -0.6624 & -2.5654\\
 -0.3289 & -0.6624 & 1.0624 & -0.0210\\
 -0.7628 & -2.5654 & -0.0210 & 3.9751
\end{bmatrix},\\
P_{3}&=\begin{bmatrix}
 2.5458 & -0.2577 & -0.4802 & -0.8349\\
 -0.2577 & 6.4729 & -0.6880 & -2.6528\\
 -0.4802 & -0.6880 & 1.1266 & 0.1856\\
 -0.8349 & -2.6528 & 0.1856 & 3.6621
\end{bmatrix},
\end{aligned}
$}
\end{equation*}
and the default controller gains
\begin{equation*}
\resizebox{0.68\hsize}{!}{$
\begin{aligned}
K_{1}&=\begin{bmatrix}
 -0.5485 & 1.4088 & 0.1241 & -1.1827\\
 0.5403 & -1.9166 & 0.1856 & 0.6356
\end{bmatrix},\\
K_{2}&=\begin{bmatrix}
 -0.9037 & 1.9401 & 0.5662 & -2.1235\\
 1.3998 & 0.0509 & -1.5069 & -1.0168
\end{bmatrix},\\
K_{3}&=\begin{bmatrix}
 1.6078 & -2.8021 & -0.4213 & 0.8067\\
 1.0113 & -0.6582 & 0.9367 & -1.8909
\end{bmatrix}.
\end{aligned}
$}
\end{equation*}

When detecting attack flow $\tilde{R}(k)$ at time $k$, the bandwidth allocation $W(k)$ and controller gain $K_{i}(k)$ are optimized jointly online by solving (\ref{PPLS-opt}), otherwise the default controller gain is adopted. The trajectory of bandwidth $W(k)$ against input flow $R+\tilde{R}(k)$ is presented in Fig.~\ref{PPLS_bandwidth_flow}. Noticing the data at $k=1,5,10$, under identical attack flow $\tilde{R}(k)=[5,5,5,5]^{\mathrm{T}}$, channels 3 and 4 are allocated with enough bandwidth preferentially for subsystem 1, channels 2 and 4 are preferred for subsystem 2, and channels 1 and 2 are preferred for subsystem 3. This result indicates the adaptivity of the strategy to not only attack flow but also system dynamics.

\begin{figure}[htb]
   \centering
   \includegraphics[width=.5\textwidth]{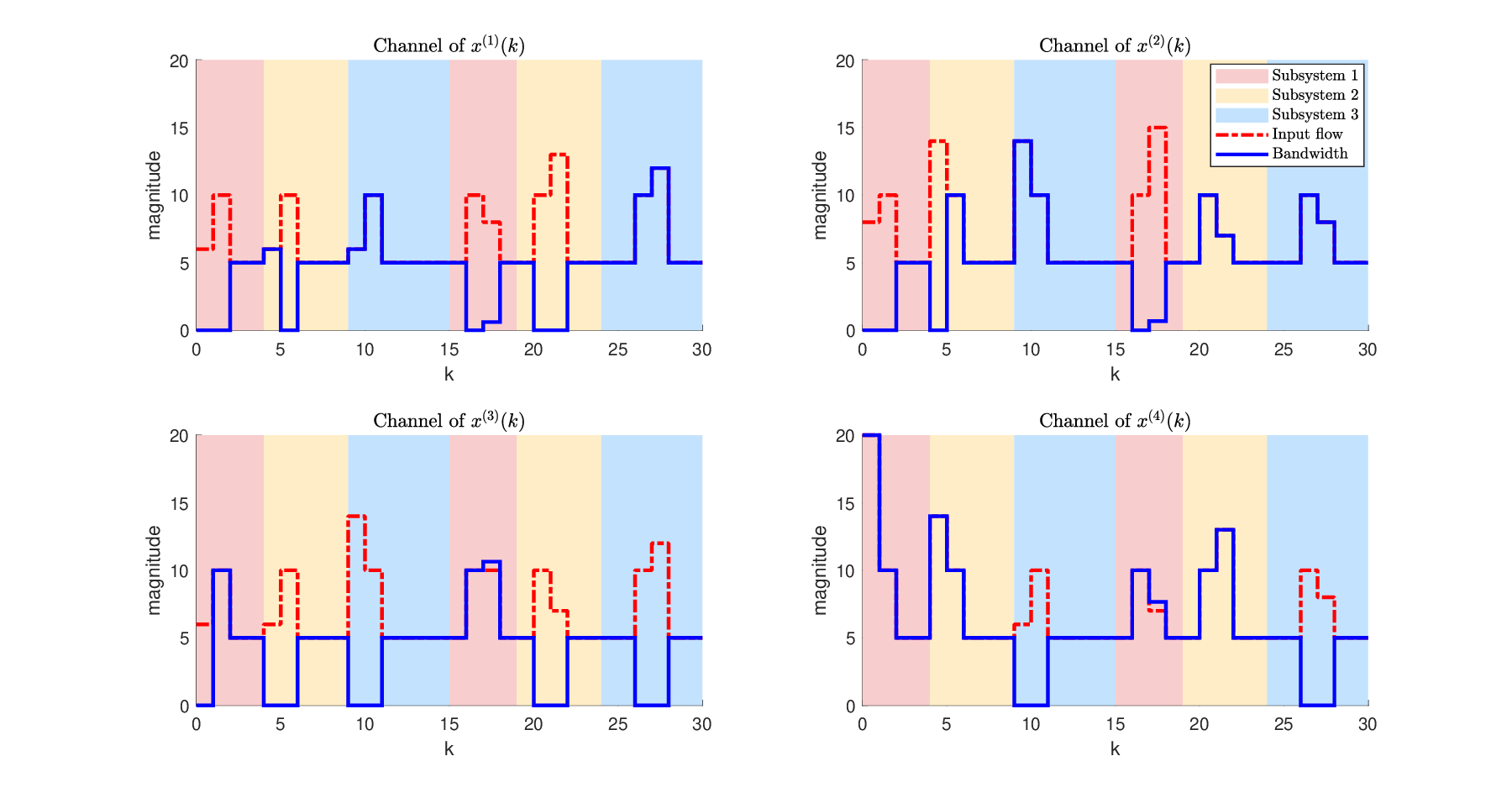}
   \caption{Trajectory of bandwidth against input flow}
   \label{PPLS_bandwidth_flow}
\end{figure}

By solving (\ref{PPLS-opt-bi}) with SEA, we have $\bar{\beta}_{1}=1.5038$, $\bar{\beta}_{2}=3.1578$ and $\bar{\beta}_{3}=3.4006$. Then we can obtain that system $\Sigma$ is 0.9520-exponentially stable with Theorem \ref{PPLS-theo-stability}. It is consistent with the closed-loop state response shown in Fig.~\ref{PPLS_closed}.

To highlight the advantage of our cross-layered strategy, two additional strategies are designed for comparison. In Strategy A, given bandwidth allocation $W^{(j)}(k)=R^{(j)}$ ($j=1,2,3,4$), only the controller gain is regulated in Problem (\ref{PPLS-opt}), the state response is shown in Fig.~\ref{PPLS_closed_onlyK}. In Strategy B, given controller gain from Lemma \ref{PPLS-lemma-alpha}, only bandwidth allocation is regulated in Problem (\ref{PPLS-opt}), the state response is shown in Fig.~\ref{PPLS_closed_onlyW}. Compared with Fig.~\ref{PPLS_closed}, it can be seen that the cross-layered strategy has advantage in transient performance with reduced overshoot and oscillation.

\begin{figure}[htb]
   \centering
   \includegraphics[width=.43\textwidth]{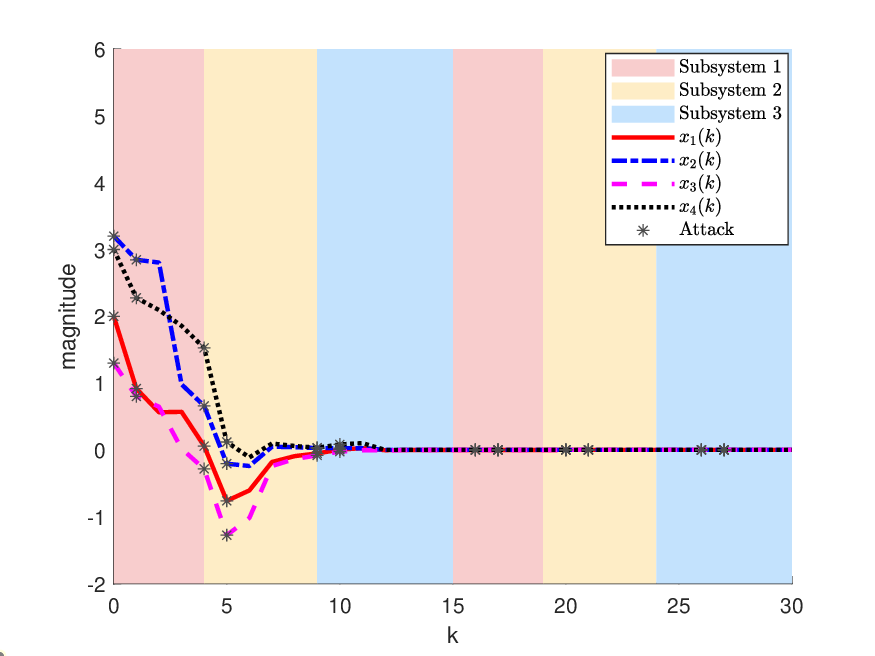}
   \caption{State response under cross-layered strategy}
   \label{PPLS_closed}
\end{figure}

\begin{figure}[htb]
   \centering
   \includegraphics[width=.43\textwidth]{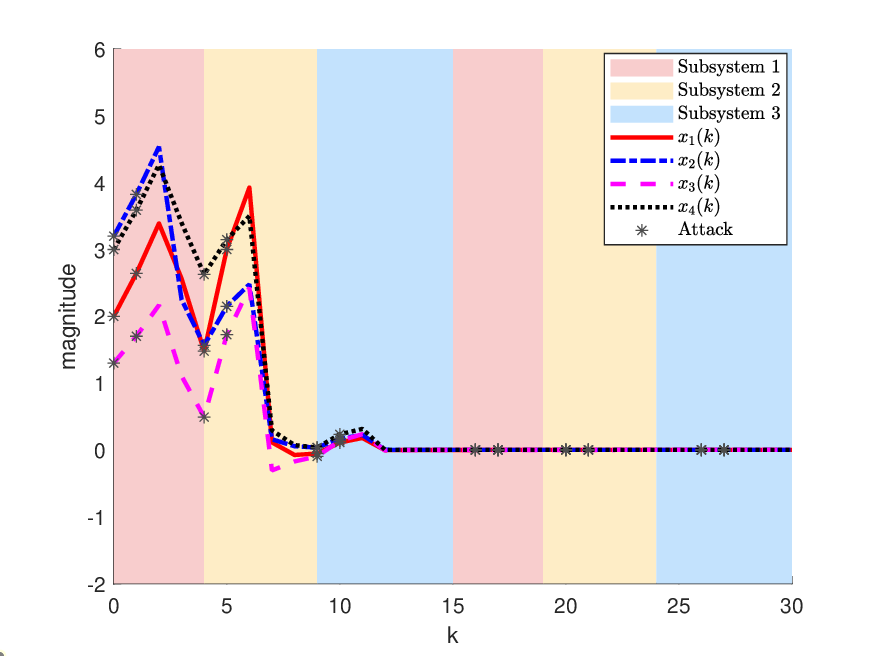}
   \caption{State response under Strategy A (under only controller gain regulation)}
   \label{PPLS_closed_onlyK}
\end{figure}

\begin{figure}[htb]
   \centering
   \includegraphics[width=.43\textwidth]{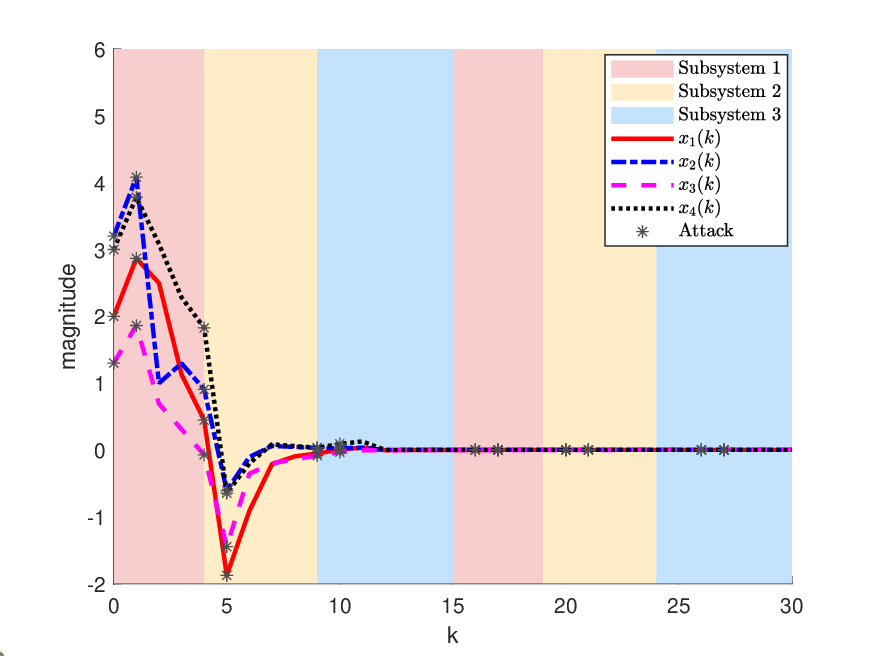}
   \caption{State response under Strategy B (under only bandwidth allocation regulation)}
   \label{PPLS_closed_onlyW}
\end{figure}

To illustrate how to employ SEA to obtain $\bar{\beta}_{i}$ for (\ref{PPLS-opt-bi}), taking subsystem 1 as an example, we present the detailed intermediate results as follows. In \textbf{Steps 1,\:2}, by solving (\ref{PPLS-opt}) with given $L_{1}$ from each $\mathcal{L}\in \mathcal{S}_{all}$ one by one, we have the optimal exponential order $\beta^{*}_{1}\mid\mathcal{L}$ for each channels' state $\mathcal{L}$, which are shown in Table \ref{tab-S-all}.
\begin{table*}
\small
\centering
\caption{Optimal exponential orders of channels' states}
\label{tab-S-all}
\begin{tabular}{ccccccccc}
\toprule
$\mathcal{L}\in \mathcal{S}_{all}$ & [1\:1\:1\:1] & [1\:0\:1\:1] & [0\:1\:1\:1] & [0\:0\:1\:1] & [1\:1\:0\:1] & [0\:1\:0\:1] & [1\:0\:0\:1] & [0\:0\:0\:1]\\ \midrule
$\beta_{1}^{*}\mid\mathcal{L}$ & 1.2689 & 1.2689 & 1.2689 & 1.2689 & 1.3737 & 1.3926 & 1.4258 & 1.4275\\
\toprule
$\mathcal{L}\in \mathcal{S}_{all}$ & [1\:1\:1\:0] & [0\:1\:1\:0] & [1\:1\:0\:0] & [0\:1\:0\:0] & [1\:0\:1\:0] & [1\:0\:0\:0] & [0\:0\:1\:0] & [0\:0\:0\:0]\\ \midrule
$\beta_{1}^{*}\mid\mathcal{L}$ & 1.5038 & 1.5646 & 1.6701 & 1.7068 & 1.8282 & 1.9140 & 2.0299 & 2.0661\\
\toprule
\end{tabular}
\end{table*}

In \textbf{Steps 3,\:4}, according to Table \ref{tab-S-all}, with (\ref{S-force}) and (\ref{L-force}), we can determine $\mathcal{S}_{force}$, which denotes the set of the channels' states that the attacker can reach forcibly by violating the delay condition $\big(R^{(j)}+\tilde{R}^{(j)}\big)\tau< S^{(j)}$ in (\ref{PPLS-opt-bi-logic}). The result is shown in Table \ref{tab-S-force}, where channels $\mathcal{L}_{f}(\mathbf{0})$ are jammed forcibly by violating the delay condition, `?' denotes the channel state dependent on the following bandwidth allocation, and \textbf{the bold entry indicates} $\tilde{\beta}_{1}^{*}=\max\limits_{\mathcal{L}_{f}\in \mathcal{S}_{force}}\left\{\beta_{1}^{*}\mid\mathcal{L}_{f}\right\}$.
\begin{table}[H]
\footnotesize
\centering
\caption{Optimal exponential orders of $\mathcal{S}_{force}$}
\label{tab-S-force}
\begin{tabular}{cccccc}
\toprule
$\mathcal{L}_{f}\in \mathcal{S}_{force}$ & \text{[0\:?\:?\:?]} & \text{[?\:0\:?\:?]} & \text{[?\:?\:0\:?]} & \text{[?\:?\:?\:0]} & \text{[?\:?\:?\:?]}\\ \midrule
$\beta_{1}^{*}\mid \mathcal{L}_{f}$ & 1.2689 & 1.2689 & 1.3737 & \textbf{1.5038} & 1.2689\\
\toprule
\end{tabular}
\end{table}

In \textbf{Steps 5,\:6}, with (\ref{S-safe}) and (\ref{L-safe}), we have $\mathcal{S}_{safe}\mid\mathcal{L}_{f}$, which denotes the set of the channels' states that the defender can guarantee by allocating bandwidth under a given $\mathcal{L}_{f}$, that is, when the attack flow jams channels $\mathcal{L}_{f}(\mathbf{0})$ by violating the delay condition. The detailed result is shown in Table \ref{tab-S-safe}, where `--' denotes the entry which can take the value of $0$ or $1$. For example, [0\:--] denotes \{[0\:0],\:[0\:1]\}. With Table \ref{tab-S-all} and $\hat{\beta}_{i}^{*}\mid\mathcal{L}_{f}=\min\limits_{\mathcal{L}_{s}\in\mathcal{S}_{safe}\mid\mathcal{L}_{f}} \left\{\beta_{i}^{*}\mid\mathcal{L}_{s}\right\}$ defined in \textbf{Step 5}, one has
\begin{equation*}
\begin{aligned}
\hat{\beta}_{1}^{*}\mid \text{[0\:?\:?\:?]}&=\beta_{1}^{*}\mid \text{[0\:1\:1\:1]}=1.2689,\\
\hat{\beta}_{1}^{*}\mid \text{[?\:0\:?\:?]}&=\beta_{1}^{*}\mid \text{[1\:0\:1\:1]}=1.2689,\\
\hat{\beta}_{1}^{*}\mid \text{[?\:?\:0\:?]}&=\beta_{1}^{*}\mid \text{[1\:1\:0\:1]}=1.3737,\\
\hat{\beta}_{1}^{*}\mid \text{[?\:?\:?\:0]}&=\beta_{1}^{*}\mid \text{[1\:1\:1\:0]}=1.5038.
\end{aligned}
\end{equation*}
Especially since $\mathcal{S}_{safe}\mid \text{[?\:?\:?\:?]}=\text{[0\:0\:0\:0]}$, one has
\begin{equation*}
\hat{\beta}_{1}^{*}\mid \text{[?\:?\:?\:?]}=\beta_{1}^{*}\mid\text{[0\:0\:0\:0]}=2.0661.
\end{equation*}

Having determined $\tilde{\beta}_{1}^{*}$ and $\hat{\beta}_{1}^{*}\mid \mathcal{L}_{f}$, one has $\mathcal{S}_{opt}\mid \mathcal{L}_{f}$ according to (\ref{S-opt-Lf}). In detail, one has
\begin{equation*}
\begin{aligned}
&\mathcal{S}_{opt}\mid \text{[0\:?\:?\:?]}\\
=&\{\mathcal{L}_{o}\in \mathcal{S}_{all}\mid \text{[0\:?\:?\:?]}:1.5038\leq \beta_{1}^{*}\mid \mathcal{L}_{o}\leq 1.2689\}=
\emptyset.
\end{aligned}
\end{equation*}
It implies that the global optimum $\tilde{R}$ of Problem (\ref{PPLS-opt-bi}) should not be under $\mathcal{L}_{f}=\text{[0\:?\:?\:?]}$, that is, the global optimum $\tilde{R}$ should not jam channel 1 forcibly by violating the delay condition. Hence, we can let $\bar{\beta}_{1}\mid \text{[0\:?\:?\:?]}=-\infty$. Similarly, we let $\bar{\beta}_{1}\mid \text{[?\:0\:?\:?]}=-\infty$ and $\bar{\beta}_{1}\mid \text{[?\:?\:0\:?]}=-\infty$, since $\mathcal{S}_{opt}\mid \text{[?\:0\:?\:?]}=\emptyset$ and $\mathcal{S}_{opt}\mid \text{[?\:?\:0\:?]}=\emptyset$. On the other hand, one has
\begin{equation*}
\begin{aligned}
&\mathcal{S}_{opt}\mid \text{[?\:?\:?\:0]}\\
=&\{\mathcal{L}_{o}\in \mathcal{S}_{all}\mid \text{[?\:?\:?\:0]}:1.5038\leq \beta_{1}^{*}\mid \mathcal{L}_{o}\leq 1.5038\}\\
=&\text{[1\:1\:1\:0]},
\end{aligned}
\end{equation*}
and
\begin{equation*}
\begin{aligned}
&\mathcal{S}_{opt}\mid \text{[?\:?\:?\:?]}\\
=&\{\mathcal{L}_{o}\in \mathcal{S}_{all}\mid \text{[?\:?\:?\:?]}:1.5038\leq \beta_{1}^{*}\mid \mathcal{L}_{o}\leq 2.0661\},
\end{aligned}
\end{equation*}
of which the elements are listed in the 3rd row of Table \ref{tab-S-all}. By running \textbf{Step 6} for $\mathcal{S}_{opt}\mid\text{[?\:?\:?\:0]}$ and $\mathcal{S}_{opt}\mid\text{[?\:?\:?\:?]}$ respectively, one has $\bar{\beta}_{1}\mid\text{[?\:?\:?\:0]}=1.5038$ and $\bar{\beta}_{1}\mid\text{[?\:?\:?\:?]}=1.5038$.

\begin{table}[H]
\footnotesize
\centering
\caption{Elements of $\mathcal{S}_{safe}\mid\mathcal{L}_{f}$}
\label{tab-S-safe}
\begin{tabular}{cccccc}
\toprule
$\mathcal{L}_{f}$ & \text{[0\:?\:?\:?]} & \text{[?\:0\:?\:?]} & \text{[?\:?\:0\:?]} & \text{[?\:?\:?\:0]} & \text{[?\:?\:?\:?]}\\ \midrule
$\mathcal{S}_{safe}\mid\mathcal{L}_{f}$ & \text{[0\:--\:--\:--]} & \text{[--\:0\:--\:--]} & \text{[--\:--\:0\:--]} & \text{[--\:--\:--\:0]} & \text{[0\:0\:0\:0]}\\
\toprule
\end{tabular}
\end{table}

In \textbf{Step 7}, we can obtain $\bar{\beta}_{1}=\max\limits_{\mathcal{L}_{f}\in \mathcal{S}_{force}}\left\{\bar{\beta}_{1}\mid\mathcal{L}_{f}\right\}=1.5038$. With regard to computational complexity, it can be seen that after the filtering in \textbf{Steps 4,\:5}, the total number of elements in $\mathcal{S}_{opt}\mid \mathcal{L}_{f}$ ($\mathcal{L}_{f}\in \mathcal{S}_{force}$) decreases from $4\times 2^{3}+2^{4}$ to $1+8$, where without the filtering of \textbf{Steps 4,\:5}, the numbers of the elements in $\mathcal{S}_{opt}\mid \text{[0\:?\:?\:?]}$, $\mathcal{S}_{opt}\mid \text{[?\:0\:?\:?]}$, $\mathcal{S}_{opt}\mid \text{[?\:?\:0\:?]}$, $\mathcal{S}_{opt}\mid \text{[?\:?\:?\:0]}$ are $2^{3}$ identically, and the number of the elements in $\mathcal{S}_{opt}\mid \text{[?\:?\:?\:?]}$ is $2^{4}$. This reduction illustrates the developed efficiency of SEA.

\section{Conclusion}
In this paper, an online cross-layered defense strategy has been proposed for multi-channel PPLSs under DoS attacks. The bandwidth allocation and controller gain are optimized jointly according to the real-time attack flow and system dynamics, by solving a mixed-integer semidefinite programming online. Furthermore, the stability of the system under the proposed strategy has been analyzed, by solving a bi-level mixed-integer semidefinite programming. To solve this non-convex bi-level programming, a smart enumeration algorithm has been proposed with reduced computation workload. The effectiveness of the strategy has been illustrated by numerical examples.

\section*{References}
\bibliographystyle{IEEEtran}
\bibliography{ref}

% Generated by IEEEtran.bst, version: 1.14 (2015/08/26)
\begin{thebibliography}{10}
\providecommand{\url}[1]{#1}
\csname url@samestyle\endcsname
\providecommand{\newblock}{\relax}
\providecommand{\bibinfo}[2]{#2}
\providecommand{\BIBentrySTDinterwordspacing}{\spaceskip=0pt\relax}
\providecommand{\BIBentryALTinterwordstretchfactor}{4}
\providecommand{\BIBentryALTinterwordspacing}{\spaceskip=\fontdimen2\font plus
\BIBentryALTinterwordstretchfactor\fontdimen3\font minus
  \fontdimen4\font\relax}
\providecommand{\BIBforeignlanguage}[2]{{%
\expandafter\ifx\csname l@#1\endcsname\relax
\typeout{** WARNING: IEEEtran.bst: No hyphenation pattern has been}%
\typeout{** loaded for the language `#1'. Using the pattern for}%
\typeout{** the default language instead.}%
\else
\language=\csname l@#1\endcsname
\fi
#2}}
\providecommand{\BIBdecl}{\relax}
\BIBdecl

\bibitem{network}
J.~F. Kurose and K.~W. Ross, \emph{{Computer Networking: {A} Top-down
  Approach}}, 8th~ed.\hskip 1em plus 0.5em minus 0.4em\relax {Boston, MA, USA}:
  {Pearson}, 2017.

\bibitem{review-Annual-2019}
S.~M. Dibaji, M.~Pirani, D.~B. Flamholz, A.~M. Annaswamy, K.~H. Johansson, and
  A.~Chakrabortty, ``A systems and control perspective of {CPS} security,''
  \emph{Annu. Rev. Control}, vol.~47, pp. 394--411, 2019.

\bibitem{review-CJAS-2022}
W.~Duo, M.~Zhou, and A.~Abusorrah, ``A survey of cyber attacks on cyber
  physical systems: {Recent} advances and challenges,'' \emph{IEEE/CAA J.
  Autom. Sinica}, vol.~9, no.~5, pp. 784--800, May 2022.

\bibitem{review-Annual-2023}
G.~Li, L.~Ren, Y.~Fu, Z.~Yang, V.~Adetola, J.~Wen, Q.~Zhu, T.~Wu, K.~Candan,
  and Z.~O'Neill, ``A critical review of cyber-physical security for building
  automation systems,'' \emph{Annu. Rev. Control}, vol.~55, pp. 237--254, 2023.

\bibitem{res-Italy-2015}
C.~De~Persis and P.~Tesi, ``Input-to-state stabilizing control under
  denial-of-service,'' \emph{IEEE Trans. Autom. Control}, vol.~60, no.~11, pp.
  2930--2944, Nov. 2015.

\bibitem{deter-Wang-2022}
Y.-W. Wang, Z.-H. Zeng, X.-K. Liu, and Z.-W. Liu, ``Input-to-state stability of
  switched linear systems with unstabilizable modes under {DoS} attacks,''
  \emph{Automatica}, vol. 146, p. 110607, Dec. 2022.

\bibitem{deter-Zhao-2022}
R.~Zhao, Z.~Zuo, and Y.~Wang, ``Event-triggered control for switched systems
  with denial-of-service attack,'' \emph{IEEE Trans. Autom. Control}, vol.~67,
  no.~8, pp. 4077--4090, Aug. 2022.

\bibitem{deter-Xu-2024}
K.~Xu, Y.~Niu, and J.~Lam, ``Secure decentralized event-triggered load
  frequency control design for multiarea power systems under multiple {{DoS}}
  attacks,'' \emph{IEEE Trans. Cybern.}, vol.~54, no.~11, pp. 6423--6435, Nov.
  2024.

\bibitem{active-2021}
T.~Li, B.~Chen, L.~Yu, and W.-A. Zhang, ``Active security control approach
  against {DoS} attacks in cyber-physical systems,'' \emph{IEEE Trans. Autom.
  Control}, vol.~66, no.~9, pp. 4303--4310, Sep. 2021.

\bibitem{active-2023}
R.~Zhao, Z.~Zuo, Y.~Wang, and W.~Zhang, ``Active control strategy for switched
  systems against asynchronous {{DoS}} attacks,'' \emph{Automatica}, vol. 148,
  p. 110765, Feb. 2023.

\bibitem{Bern-Li-2024}
Y.~Li, H.~Lin, C.~Zhao, and J.~Lam, ``Centralized approximate optimal
  estimation for cyber-physical systems under joint cyber-attacks,'' \emph{Int.
  J. Robust Nonlinear Control}, vol.~34, no.~5, pp. 3297--3317, 2024.

\bibitem{Bern-TAC-2025}
C.~Tan, J.~Di, G.~Guo, Y.~Li, and W.~S. Wong, ``Exponential mean-square
  stabilization control for cyber-physical systems under random {{DoS}} attacks
  and transmission delay,'' \emph{IEEE Trans. Autom. Control}, vol.~70, no.~1,
  pp. 190--202, Jan. 2025.

\bibitem{res-Sun-2017}
H.~Sun, C.~Peng, T.~Yang, H.~Zhang, and W.~He, ``Resilient control of networked
  control systems with stochastic denial of service attacks,''
  \emph{Neurocomputing}, vol. 270, pp. 170--177, Dec. 2017.

\bibitem{Markov-Xue-2025}
M.~Xue, J.~Lam, H.~Yan, and K.-W. Kwok, ``Asynchronous control for interval
  {{Type-2}} fuzzy nonhomogeneous {{Markov}} jump systems against successive
  {{DoS}} attacks,'' \emph{IEEE Trans. Cybern.}, vol.~55, no.~1, pp. 172--183,
  Jan. 2025.

\bibitem{Optimal_DoS_1}
H.~Zhang, P.~Cheng, L.~Shi, and J.~Chen, ``Optimal denial-of-service attack
  scheduling with energy constraint,'' \emph{IEEE Trans. Autom. Control},
  vol.~60, no.~11, pp. 3023--3028, Nov. 2015.

\bibitem{Optimal_DoS_2}
------, ``Optimal {DoS} attack scheduling in wireless networked control
  system,'' \emph{IEEE Trans. Control Syst. Technol.}, vol.~24, no.~3, pp.
  843--852, May 2016.

\bibitem{bandwidth-limit-conf-2018}
Z.~Wang and Y.~Cheng, ``Bandwidth allocation strategy of networked control
  system under denial-of-service attack,'' in \emph{Proc. Int. Conf. Netw. Inf.
  Syst. Comput.}\hskip 1em plus 0.5em minus 0.4em\relax {Wuhan, China}: {IEEE},
  Apr. 2018, pp. 49--55.

\bibitem{bandwidth-limit-Hossain-2022}
M.~M. Hossain, C.~Peng, H.-T. Sun, and S.~Xie, ``Bandwidth allocation-based
  distributed event-triggered {LFC} for smart grids under hybrid attacks,''
  \emph{IEEE Trans. Smart Grid}, vol.~13, no.~1, pp. 820--830, Jan. 2022.

\bibitem{bandwidth-redundant-Hu-2021}
S.~Hu, Z.~Cheng, D.~Yue, C.~Dou, and Y.~Xue, ``Bandwidth allocation-based
  switched dynamic triggering control against {DoS} attacks,'' \emph{IEEE
  Trans. Syst., Man, Cybern., Syst.}, vol.~51, no.~10, pp. 6050--6061, Oct.
  2021.

\bibitem{PPLS-Auto-2015}
P.~Li, J.~Lam, and K.~C. Cheung, ``Stability, stabilization and {$L_{2}$}-gain
  analysis of periodic piecewise linear systems,'' \emph{Automatica}, vol.~61,
  pp. 218--226, Nov. 2015.

\bibitem{PPLS-Auto-2018}
P.~Li, J.~Lam, K.-W. Kwok, and R.~Lu, ``Stability and stabilization of periodic
  piecewise linear systems: {{A}} matrix polynomial approach,''
  \emph{Automatica}, vol.~94, pp. 1--8, Aug. 2018.

\bibitem{PPLS-TAC-2019}
P.~Li, J.~Lam, R.~Lu, and K.-W. Kwok, ``Stability and {$L_{2}$} synthesis of a
  class of periodic piecewise time-varying systems,'' \emph{IEEE Trans. Autom.
  Control}, vol.~64, no.~8, pp. 3378--3384, Aug. 2019.

\bibitem{PPLS-JFI-2022}
L.~Wan, P.~Li, P.~Li, and M.~Xing, ``Stabilization of networked periodic
  piecewise linear systems under asynchronous switching with packet dropouts,''
  \emph{J. Franklin Inst.}, vol. 359, no.~14, pp. 7733--7752, Sep. 2022.

\bibitem{PPLS-almost-2023}
C.~Fan, J.~Lam, X.~Xie, and P.~Li, ``Stability and stabilization of almost
  periodic piecewise linear systems with dwell time uncertainty,'' \emph{IEEE
  Trans. Autom. Control}, vol.~68, no.~2, pp. 1130--1137, Feb. 2023.

\bibitem{exponential_stability}
A.~Halanay and V.~Rasvan, \emph{{Stability and Stable Oscillations in Discrete
  Time Systems}}, 1st~ed.\hskip 1em plus 0.5em minus 0.4em\relax Boca Raton,
  FL, USA: {CRC Press}, 2000.

\bibitem{res-Feng-2017}
S.~Feng and P.~Tesi, ``Resilient control under denial-of-service: {Robust}
  design,'' \emph{Automatica}, vol.~79, pp. 42--51, May 2017.

\bibitem{multi-Lu-2018}
A.-Y. Lu and G.-H. Yang, ``Input-to-state stabilizing control for
  cyber-physical systems with multiple transmission channels under denial of
  service,'' \emph{IEEE Trans. Autom. Control}, vol.~63, no.~6, pp. 1813--1820,
  Jun. 2018.

\bibitem{opt-convex}
S.~P. Boyd and L.~Vandenberghe, \emph{{Convex Optimization}}, 1st~ed.\hskip 1em
  plus 0.5em minus 0.4em\relax {Cambridge, UK}: {Cambridge University Press},
  2004.

\bibitem{opt-logic}
H.~P. Williams, \emph{{Model Building in Mathematical Programming}},
  5th~ed.\hskip 1em plus 0.5em minus 0.4em\relax {West Sussex, UK}: {Wiley},
  2013.

\bibitem{opt-product}
M.~Asghari, A.~M. {Fathollahi-Fard}, S.~M.~J. {Mirzapour Al-e-hashem}, and
  M.~A. Dulebenets, ``Transformation and linearization techniques in
  optimization: {A} state-of-the-art survey,'' \emph{Mathematics}, vol.~10,
  no.~2, p. 283, Jan. 2022.

\bibitem{YALMIP}
J.~Lofberg, ``{{YALMIP}}: A toolbox for modeling and optimization in
  {{MATLAB}},'' in \emph{Proc. IEEE Int. Symp. Comput. Aid. Control Syst.
  Des.}\hskip 1em plus 0.5em minus 0.4em\relax Taipei, Taiwan: IEEE, 2004, pp.
  284--289.

\bibitem{MOSEK}
\BIBentryALTinterwordspacing
{MOSEK ApS}, ``{MOSEK} version 10,'' 2024. [Online]. Available:
  \url{https://www.mosek.com}
\BIBentrySTDinterwordspacing

\end{thebibliography}

\end{document}